\documentclass[11pt]{article}

\title{Conditional risk measures  in a bipartite market  structure}

\author{Oliver Kley\thanks{Center for Mathematical Sciences, Technische Universit\"at M\"unchen,  85748 Garching, Boltzmannstrasse 3, Germany, e-mail: oliver.kley@tum.de, cklu@tum.de}
\and Claudia Kl\"uppelberg\footnotemark[1]
\and Gesine Reinert\thanks{Department of Statistics, University of Oxford, 1 South Parks Road, Oxford OX1 3TG, UK, email: reinert@stats.ox.ac.uk }
}

\usepackage{amssymb}
\usepackage{amsfonts}
\usepackage{amsmath}
\usepackage{amsthm}
\usepackage{graphicx}
\usepackage{caption}
\usepackage[usenames, dvipsnames]{xcolor}
\usepackage{verbatim}
\usepackage{dsfont}
\usepackage{color}
\usepackage{hyperref}
\usepackage{subfigure}
\usepackage{comment}
 \usepackage{tikz}

\numberwithin{equation}{section}

\newtheorem{theorem}{Theorem}[section]
\newtheorem{lemma}[theorem]{Lemma}
\newtheorem{remark}[theorem]{Remark}
\newtheorem{example}[theorem]{Example}
\newtheorem{proposition}[theorem]{Proposition}
\newtheorem{definition}[theorem]{Definition}

\newtheorem{corollary}[theorem]{Corollary}
\newtheorem{fig}[theorem]{Figure}

\newcommand{\bthe}{\begin{theorem}}
\newcommand{\ethe}{\end{theorem}}

\newcommand{\ben}{\begin{enumerate}}
\newcommand{\een}{\end{enumerate}}

\newcommand{\bit}{\begin{itemize}}
\newcommand{\eit}{\end{itemize}}

\newcommand{\beq}{\begin{equation}}
\newcommand{\eeq}{\end{equation}}

\newcommand{\ble}{\begin{lemma}}
\newcommand{\ele}{\end{lemma}}

\newcommand{\bde}{\begin{definition}\rm}
\newcommand{\ede}{\Chalmos\end{definition}}

\newcommand{\bco}{\begin{corollary}}
\newcommand{\eco}{\end{corollary}}

\newcommand{\bpr}{\begin{proposition}}
\newcommand{\epr}{\end{proposition}}

\newcommand{\brem}{\begin{remark}\rm}

\newcommand{\erem}{\Chalmos\end{remark}}

\newcommand{\bproof}{\begin{proof}}
\newcommand{\eproof}{\end{proof}}

\newcommand{\bexam}{\begin{example}\rm}

\newcommand{\eexam}{\Chalmos\end{example}}

\newcommand{\bfi}{\begin{fig}}
\newcommand{\efi}{\end{fig}}

\newcommand{\btab}{\begin{tab}}
\newcommand{\etab}{\end{tab}}

\newcommand{\beao}{\begin{eqnarray*}}
\newcommand{\eeao}{\end{eqnarray*}\noindent}

\newcommand{\beam}{\begin{eqnarray}}
\newcommand{\eeam}{\end{eqnarray}\noindent}

\newcommand{\overliner}{\begin{array}}
\newcommand{\earr}{\end{array}}

\newcommand{\bdis}{\begin{displaymath}}
\newcommand{\edis}{\end{displaymath}\noindent}

\def\N{{\mathbb N}}

\def\P{{\mathbb P}}
\def\E{{\mathbb E}}
\def\R{{\mathbb R}}

\newcommand{\pr}[1]{\P\left(#1\right)}
\def\P{{\mathbb P}}

\def\calb{{\mathcal{B}}}

\def\calr{{\mathcal{R}}}

\def\cale{{\mathcal{E}}}

\def\calw{{\mathcal{W}}}

\def\1{\mathds{1}}
\newcommand{\bone}{\mathds{1}}
\newcommand{\bnull}{{0}}

\newcommand{\stv}{\stackrel{v}{\rightarrow}}

\newcommand{\tto}{{t\to\infty}}

\newcommand{\al}{{\alpha}}

\newcommand{\la}{{\lambda}}

\newcommand{\ga}{{\gamma}}

\newcommand{\si}{{\sigma}}

\newcommand{\eps}{\varepsilon}

\newcommand{\VaR}{{\rm VaR}}
\newcommand{\ICoVaR}{{\rm ICoVaR}}
\newcommand{\SCoVaR}{{\rm SCoVaR}}
\newcommand{\CoTE}{{\rm CoTE}}
\newcommand{\ICoTE}{{\rm ICoTE}}
\newcommand{\SCoTE}{{\rm SCoTE}}

\newcommand{\MCoVaR}{{\rm MCoVaR}}
\newcommand{\MCoTE}{{\rm MCoTE}}

\newcommand{\dep}{{\rm dep}}
\newcommand{\ind}{{\rm ind}}
\newcommand{\AS}{{\rm AS}}
\newcommand{\SA}{{\rm SA}}

\newcommand{\Pois}{{\rm Pois}}

\newcommand{\ov}{\overline}

\newcommand{\Chalmos}{\quad\hfill\mbox{$\Box$}}  

\newcommand{\CK}[1]{{\color{blue} #1}}
\newcommand{\OK}[1]{{\color{BrickRed} #1}}
\newcommand{\GR}[1]{{\color{Brown} #1}}

\allowdisplaybreaks[4]

\begin{document}


\maketitle

\begin{abstract}
In this paper we study the effect of network structure between agents and objects on measures for systemic risk.
We model the influence of sharing large exogeneous losses to the financial or (re)insuance market by a bipartite graph. Using Pareto-tailed losses and multivariate regular variation we obtain asymptotic results for systemic conditional risk measures based on the Value-at-Risk and the Conditional Tail Expectation. These results allow us to assess the influence of an individual institution on the systemic or market risk and vice versa through a collection of conditional systemic risk measures. For large markets Poisson approximations of the relevant constants are provided in the example of an insurance market. The example of an  underlying homogeneous random graph is analysed in detail, and the results are illustrated through simulations.
\end{abstract}

\noindent
\begin{tabbing}
{\em MSC2010 Subject Classifications:} \= primary:\,\,\,90B15\,\,\,
secondary: \,\,\,91B30, 60G70, 62P05,  62E20
\end{tabbing}

\section{Introduction}\label{s1}

Quantitative assessments of financial risk and of (re)insurance risk has to take the interwoven web of agents and business relationships into account in order to capture systemic risk phenomena. Measuring such risks while accounting for this complex system of agents  is an ongoing area of research, see for example \cite{coherent, axiomsystemic, hannes, Huang, jouini, overbeck}. This paper joins the discussion by adapting  conditional systemic risk measures which are based on similar asymptotic arguments as classical risk measures. Making use of results derived in \cite{KKR1}, we illustrate these risk measures on a bipartite graph model for the agent-object market structure, combined with a heavy-tailed loss distribution.

The conditional systemic risk measures  in this paper are conditional versions of
the {\em Value-at-Risk} (\VaR) defined for a random variable $X$ at confidence level $1-\gamma$ as
\begin{gather*}
\VaR_{1-\gamma}(X):=\inf\{ t \geq 0: \pr{ X>t } \leq \gamma \},\quad \gamma \in (0,1),
\end{gather*}
 and the {\em Conditional Tail Expectation} (\CoTE), {also known as {\em Expected Shortfall}}, at confidence level $1-\gamma$, based on the corresponding \VaR,  as
\begin{gather}\label{CoT}
\CoTE_{1-\gamma}(X):=\E[X \mid X> \VaR_{1-\gamma}(X)],\quad \gamma \in (0,1).
\end{gather}

For a systemic risk approach it is of interest to quantify  not only the risk of single agents, but also the market risk,  which is of high relevance to regulators.
Moreover, it is natural to investigate an agent's risk based on the aggregated market risk; see e.g. Theorem~2.4 of \cite{ZhuLi}.
Consequently, we will study conditional systemic risk measures where the conditioning event involves the whole market risk {as well as} its influence on one specific agent.
In the same way, it is of interest to evaluate the market risk conditioned on the event that one agent faces high losses.
Such ideas lead to a classification of conditional systemic risk measures as in Table~\ref{table} (motivated by \cite{DJVZ}) which will be defined in Definition~\ref{def:systemicriskmeasures}. Note that the definition of CoVaR is already present in \cite{CoVar} and the $\ICoTE$ goes back to the so-called Marginal expected Shortfall from \cite{Brownlees}.   \\

\begin{table}[ht]
\begin{center}
\begin{tabular}{cccc }
marginal risk measure & institution  $|$ institution & institution $|$ system  & system  $|$ institution \\
\hline
 VaR& MCoVaR & ICoVaR &SCoVaR\\
 \hline
 CoTE &  MCoTE &  ICoTE& SCoTE\\
\hline\\
\end{tabular}
\caption{\label{table} Classifying conditional systemic risk measures: ``M'' stands for {\em mutual} indicating the risk measure of one institution given high risk in another institution; ``I'' stands for {\em individual} indicating the risk of an individual institution given high market risk; and ``S'' stands for {\em system} indicating the risk of the system given high risk of an institution. }
\end{center}
\end{table}


In \cite{axiomsystemic}, \cite{hannes} and \cite{overbeck}  an axiomatic framework for systemic risk has been suggested. This general framework assumes that a conditional systemic risk measure $\rho$ of a multivariate risk $X = (X_1,\ldots, X_n)$  can be represented as the composition of a univariate (single-agent)  risk measure $\rho_0$ with an {\it{aggregation function}} $\Lambda: \mathbb{R}^n \rightarrow \mathbb{R}^n$, so that $\rho = \rho_0 \circ \Lambda$. 
Here, $\rho_0$ is usually assumed to be convex as well as monotone and positively 1-homogeneous. While the conditions on  $\Lambda$ vary, there is consensus that $\Lambda$ should be positively 1-homogeneous, so that $\Lambda(ax) = a \Lambda(x)$ for $a > 0$. 
We deviate from \cite{axiomsystemic} in that we do not assume that $\Lambda((1,\dots,1)^{\top})=n$. Examples for such aggregation functions are $\Lambda (x) = \| x\| = \left( \sum_{i=1}^n |x_i|^r\right)^\frac{1}{r}$, which is a norm for $r \ge 1$ and a quasi-norm for $ 0 < r < 1$, and $\Lambda (x) = x_i$, the projection onto one coordinate. 
The fact that we do not require  $\Lambda((1,\dots,1)^{\top})=n$ has consequences in terms of system size: Assuming  that $\rho_0$ is monotone, the inequalities $n<\|(1,\dots,1)\|_r$ for $0<r<1$ as well as $n>\|(1,\dots,1)\|_r$ for $1<r\leq \infty$ hold. Therefore, systemic risk may increase faster or increase slower, respectively, as the number number of individual risks grows compared  to systemic risk with respect to a normalized aggregation function. 
Such effects  can be realistic as a larger market may not be proportionally risky to a smaller market due to a balance of risk as is well-known for insurance portfolios.
In addition, we argue that in a small and risky market the regulator may well strive for more risk capital than the sum of risks.
Also moral hazard from the different institutions is well-known and the regulator may guard against this hazard  by choosing a conditional systemic risk measure which is larger than the sum of the individual risks in the market as a quasi-norm would imply. Whatever type of aggregation function is chosen, in practice this is an economical decision. Our framework provides considerable variability in the choice of aggregation function.

In this paper we relate market risk to individual risk in the mathematical framework  of multivariate regular variation. This framework allows us to assess conditional systemic risk measures as in Table~\ref{table} asymptotically in a precise way.

\bde[Conditional systemic risk measures] \label{def:systemicriskmeasures}
Let  $F=(F_1,\ldots,F_q)$ be the  random exposure vector and let $\|\cdot\|$ be a norm or a quasinorm.
For $\ga_i,\ga\in (0,1)$ referring to agent $i$ and the market, respectively, the conditional systemic risk measures from Table \ref{table} are defined as follows: \\[2mm]
(a) \,  {\em Individual Conditional Value-at-Risk}\\
\centerline{$\ICoVaR_{1-\gamma_i, \gamma}(F_i \mid {h(F)} ):=\inf\{ t \geq 0: \pr{ F_i>t \mid h(F)> \VaR_{1-\gamma}(h(F)) } \leq \gamma_i \},$}\\
(b) \,  {\em Systemic Conditional Value-at-Risk} \\
\centerline{$\SCoVaR_{1-\gamma,\gamma_i}(h(F)\mid { F_i} ):=\inf\{ t \geq 0: \pr{ h(F)>t \mid F_{i} > \VaR_{1-\gamma}(F_i)  } \leq \gamma \},$}\\
(c) \,  {\em Mutual Conditional Value-at-Risk}\\
\centerline{$\MCoVaR_{1-\gamma_i, \gamma_k}(F_i \mid F_k):=\inf\{ t \geq 0: \pr{ F_i>t \mid F_{k} > \VaR_{1-\gamma_k}(F_k)} \leq \gamma_i \}$,}\\
(d) \,  {\em Individual Conditional Tail Expectation}\\
\centerline{$\ICoTE_{1-\gamma}( F_{i} \mid h(F)  ) :=\E[F_{i} \mid h(F) > \VaR_{1-\gamma}(h(F))]$,}\\
(e) \, {\em Systemic Conditional Tail Expectation}\\
\centerline{$\SCoTE_{1-\gamma}(h(F) \mid F_i  ) :=\E[  h(F)  \mid F_{i} > \VaR_{1-\gamma}(F_i)]$,}\\
(f) \,  {\em Mutual Conditional Tail Expectation}\\
\centerline{$\MCoTE_{1-\gamma}(F_{i}\mid F_{k}  ) :=\E[F_{i} \mid F_{k} > \VaR_{1-\gamma}(F_k)].$}
For the risk measures (d)-(f) finite first moments of the underlying random variables are required.
\ede

To model the {complex} interaction between economic agents and objects we use a bipartite network, see Figure~\ref{fig1} for a depiction.
The network can be summarised through a  random $q\times d$ {\em weighted adjacency matrix} $A$ given by
\begin{gather}\label{eq2.2}
 A_{ij} = W_{ij}{\1 (i\sim j)},\quad \mbox{ where } \frac{0}{0}:=0,
\end{gather}
 where $W_{ij}$ are positive weights which may depend on the underlying network. 
The objects can generate large losses as portfolios in a hedge fund, for instance, or as catastrophic claims in (re)insurance. In Section~\ref{s4} we shall see that the network is of considerable importance for the asymptotic behaviour of the conditional systemic risk measures.

\begin{figure}[ht]
\vspace*{-1cm}
\begin{center}

\includegraphics[width=0.7\textwidth]{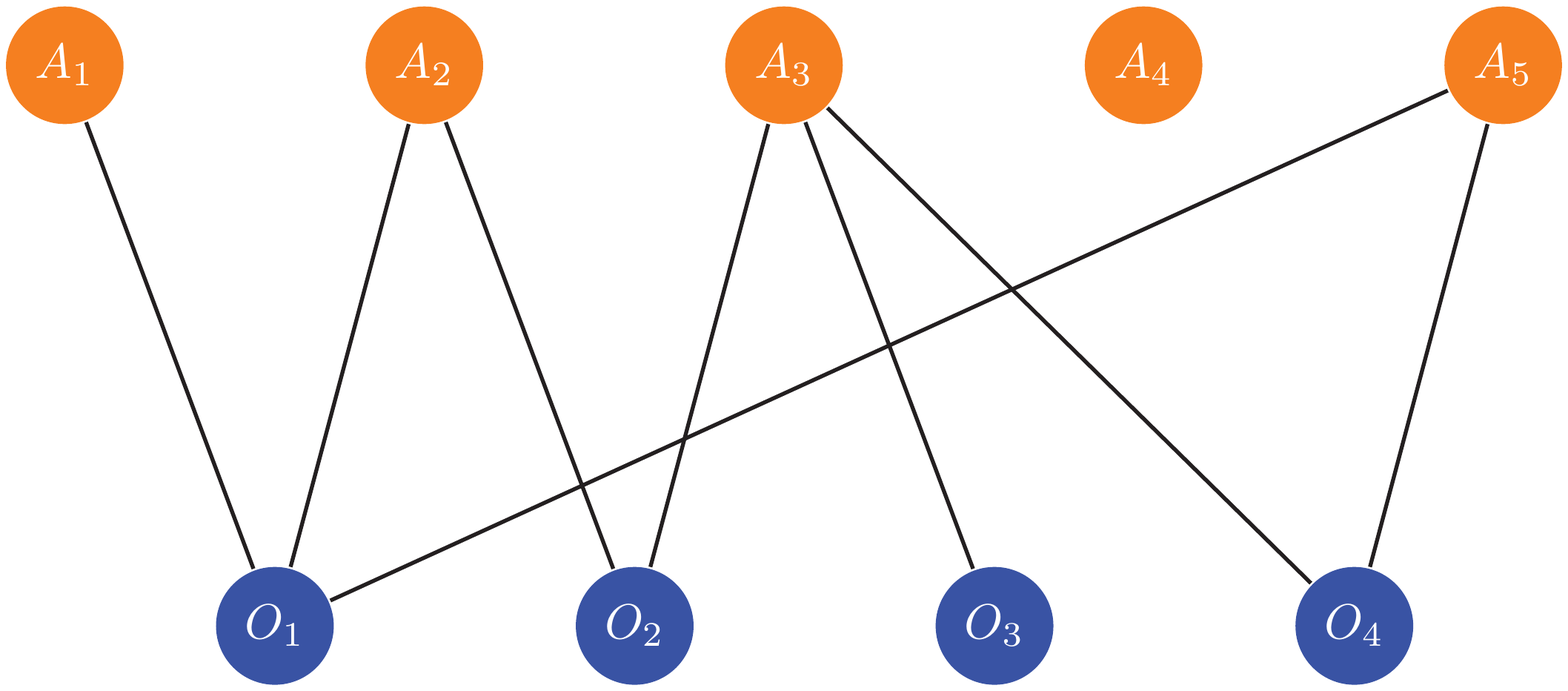}
\vspace*{-2cm}
\end{center}
\label{bipartite}
\caption{\label{fig1} The hierarchical structure of the  market as a {bipartite graph}.}
\end{figure}

Our paper is organised as follows.   In Section~\ref{s2} we formulate the bipartite graph model in detail and present the motivating examples. Section~\ref{s3} summarizes the necessary results from regular variation. Here we also present the asymptotic results of conditional probabilities and conditional expectations. While we formulate our results in the general context of regular variation with arbitrary dependence structure, we single out the two cases,  asymptotic independence and asymptotic complete dependence,  of the loss variables. In Section~\ref{s32} we discuss the asymptotic behaviour of the conditional systemic risk measures in our network model. When introducing conditional systemic risk measures, for the individual risk of every agent in the market we focus on the one-dimensional projections of the exposure vector,  and take norms and quasi-norms as appropriate aggregation functions.

 Finally, in Section~\ref{s4} we also discuss the consequence of the fact that not all claims may be insured or not all assets may find investors, respectively.  We furthermore present the homogeneous model, which exactly has this feature. Calculating the network-dependent quantities which determine the asymptotic behaviour of the conditional systemic risk measures is not always straightforward; hence we provide a Poisson approximation for some standard specifications of the model, with bounds on the total variation distance. Simulations for the homogeneous model illustrate the results.

\section{The bipartite graph model}\label{s2}

Throughout we assume that the objects, which are large claims or losses, have a random amount modelled by random variables $V_j$ for $j=1,\ldots,d$ with Pareto-tails such that, for possibly different $K_j>0$ and tail index $\al > 0$,
\beam\label{pareto}
P(V_j>t)\sim K_j t^{-\al},\quad t \to \infty.
\eeam
(For two functions $f$ and $g$ we write $f(t)\sim g(t)$ as $\tto$ if $\lim_{\tto} f(t)/g(t) = 1$.)
We summarize all objects in the vector $V=(V_1,\dots,V_d)^\top$ and assume that $V$
is independent of the random graph construction, while $V_1,\ldots,V_d$ may not be independent of each other.

Each agent may cover a random amount or proportion of an object, modelled by a random {\em weight matrix}
$W:\Omega\rightarrow \R_{+}^{q\times d}$, which satisfies the integrability condition $\E[ \|W\|^{\alpha+\delta}]< \infty$ for some matrix norm $\|\cdot\|$ and some $\delta>0$. We assume that $W_{ij}> 0 $ for all $(i,j)$ such that $i \sim j$. 
The random variable $\1 ( i\sim j )$ equals 1 whenever agent $i$ holds a contractual relationships to object $j$, and 0 otherwise.
The proportion of object $j$ which affects agent $i$ is represented by
$W_{ij} \1 ( i\sim j). $ Then $F_i:=\sum_{j=1}^d W_{ij} \1( i\sim j) $ denotes the {\em exposure of agent $i$} and
$F=(F_1,\dots,F_q)^\top$ is the vector of the joint exposures of the agents in the market. Hence, the {\em weighted adjacency  matrix} $A: \Omega \rightarrow \R^{q\times d} $  representing the market structure is given by
\begin{gather}\label{A}
A_{ij}= W_{ij}\1(i\sim j ),\quad \mbox{ where } \frac{0}{0}:=0.
\end{gather}
Consequently, the vector $F$ of agent exposures is the matrix-vector product
 \begin{gather}\label{F}
F=AV.
\end{gather}

\bexam[Large reinsurance risks, \cite{KKR1}]\label{largeclaims} In this example, 
agents are reinsurance companies and objects are large claims.
Under the simplified assumption that  claims are split into equal proportions among all agents which  insure this risk, the market matrix $A$ is 
\begin{gather}\label{eq2.2}
A_{ij} = \frac{\1 (i\sim j)}{\deg(j)},
\end{gather}
 where $\text{deg}(j)$ denotes the number of agents that insure object $j$.
 \eexam

 \bexam[Coupled portfolios of highly risky assets, \cite{Ibragimov2005}]\label{portfolios}
In this example, agents are investors and objects are investment opportunities.
Each agent $i$ has a certain amount of capital to invest, say $C_i>0.$
Again for simplicity, we assume that he splits his money in equal portions to all the assets he has chosen to invest in.
This results in a market matrix $A$  given by
\begin{gather}
A_{ij}=C_i \frac{\1 (i\sim j )}{\text{deg}(i)}
\end{gather}
where $\text{deg}(i)$ denotes the number of different assets agent $i$ invests in.
\eexam

We consider risk measures of $F=AV$, where the random matrix $A$ models the network structure of the market.
Instead of attributing a risk measure to an agent's exposure or to the market exposure, we write for short an agent's risk or the market risk.

In \cite{KKR1} it was shown
that under the assumption of regularly varying exposure vectors the asymptotic behaviour of the VaR and the CoTE can be described using the constants
\begin{equation} \label{VaRconst}
 C^i_{\ind} =  C^i_{\ind} (A) := \sum_{j=1}^{d} K_j\E A_{ij}^{\alpha},  \,\, i=1, \ldots, q,  \, \mbox{ and } \, C_{\ind}^S = C_{\ind}^S (A) = \sum_{j=1}^d K_j \E \| A e_j \|^\alpha,
\end{equation}
as well as \begin{gather}\label{VaRconst_dep}
C_\dep^i = C_\dep^i (A) :=  \E (AK^{1/\al}1)_{i}^{\alpha }, \,\, i=1, \ldots, q,   \,  \mbox{ and }\quad  C_\dep^S = C_\dep^S (A) = \E \|AK^{1/\al}1\|^\alpha,
\end{gather}
where $1$ is the $d-$dimensional vector with entries all equal 1 and $K^{1/\al}={\rm diag}(K_1^{1/\al},\ldots,K_d^{1/\al})$ is a $d\times d $ diagonal matrix. Here  the subscripts $\ind$ and $\dep$ refer to asymptotically independent or asymptotically fully dependent components of the $V_j$'s, respectively.

\ble[Corollaries~3.6 and~3.7 of \cite{KKR1}]\label{Cor:VaR}
Let $\al>0$ and $F=(F_{1},\dots,F_{q})^{\top}$ the vector of the agents' exposures.
\begin{itemize} 
\item[(a)] 
The individual Value--at--Risk of agent $i\in\{1,\ldots,q\}$ shows the asymptotic behaviour
 \begin{gather}\label{uniVaRasym}
\VaR_{1-\gamma}(F_{i})\sim C^{1/\alpha} \gamma^{-1/\alpha} ,\quad \gamma\rightarrow 0,
\end{gather}
with either $C= C^{i}_{\ind}$ or $C=C^{i}_\dep$ in case $V_1,\dots,V_d$ are asymptotically independent or asymptotically fully dependent.
The market Value--at--Risk of the aggregated vector $\|F\|$ satisfies
\begin{gather}\label{marketVaRasym}
\VaR_{1-\gamma}(\|F\|) \sim C^{1/ \alpha}\gamma^{-1/\alpha},\quad \gamma\rightarrow 0,
\end{gather}
with either $C= C^S_\ind$ or $C=C^S_\dep$ in case  $V_1,\dots,V_d$ are asymptotically independent or asymptotically fully dependent.
\item[(b)]
Let $\alpha >1.$ 
The individual Conditional Tail Expectation of  agent $i\in\{1,\ldots,n\}$ shows the asymptotic behaviour
$$\CoTE_{1-\gamma}(F_{i}) \sim \frac{\alpha}{\alpha-1} \VaR_{1-\gamma}(F_{i})\sim \frac{\alpha}{\alpha-1} C^{1/\alpha} \gamma^{-1/\alpha}   \,,\quad \gamma\rightarrow 0,$$
with either $C= C^S_{\ind}$ or $C=C^S_\dep$ in case  $V_1,\dots,V_d$ are asymptotically independent or asymptotically fully dependent.
The market Conditional Tail Expectation of  the aggregated vector $\|F\|$ satisfies
$$\CoTE_{1-\gamma}(\|F\|) \sim \frac{\alpha}{\alpha-1} \VaR_{1-\gamma}(\|F\|)\sim \frac{\alpha}{\alpha-1} C^{1/ \alpha}\gamma^{-1/\alpha}  \,,\quad \gamma\rightarrow 0,$$
with either $C= C^S_{\ind}$ or $C=C^S_\dep$ in case $V_1,\dots,V_d$ are asymptotically independent or asymptotically fully dependent.
\end{itemize} 
\ele

For the asymptotic behaviour of the Value-at-Risk and of the Conditional Tail Expectation the underlying network model enters only through the constants \eqref{VaRconst} and \eqref{VaRconst_dep}. Many underlying networks,  even networks for which the adjacency matrix is deterministic,  may hence give rise to the same asymptotic behaviour.

When the Pareto-tailed losses are independent, the constant \eqref{VaRconst} with superscript $i$ indicates the individual setting of agent $i$, whereas $S$ refers to the systemic setting.
We contrast this with the fully dependent case;
  the corresponding quantities are {given in \eqref{VaRconst_dep}}.
In general, small constants are more desirable, indicating a smaller risk.
The case of fully dependent objects is equivalent to having a single source of risk, but with the loss to be unevenly distributed among the agents.

As indicated in \cite{KKR1} these two extreme dependence cases give rise to risk bounds (cf. \cite{KK3}), which are determined using  the constants {given in \eqref{VaRconst}} and \eqref{VaRconst_dep}.

\section{Asymptotic results from multivariate regular variation}\label{s3}

To obtain asymptotic results as in Lemma~\ref{Cor:VaR} for the conditional systemic risk measures from Definition~\ref{def:systemicriskmeasures} in a more general framework, we first extend classical results for regular variation to continuous 1-homogeneous functions. Examples for such continuous 1-homogeneous functions are projections of the vector $F=(F_1,\dots,F_q)^{\top}$ on the $i$-th coordinate $F_i$, and  the norm or quasi-norm of the vector $F$,  which link up with  Section~\ref{s2}.

Our framework will be regular variation of the random vector of exposures $F$, which follows from the Pareto-tailed claims and the dependence structure introduced by the bipartite graph; cf. \cite{KKR1}.
There are several equivalent definitions of multivariate regular variation; cf. Theorem~6.1 of \cite{Resnick2007} and Ch.~2.1 of \cite{BasrakPhD}.
Also notions like {\em one point uncompactification} and {\em vague convergence} are defined there, referring to \cite{Resnick2007}, Section~6.1.3, for more background.

For $d\in\N$, let $\mathbb{S}_{+}^{d-1}=\{ x\in \R^{d}_{+}: \ \| x\|=1 \}$ denote the positive unit sphere in $\R^{d}$ with respect to an arbitrary norm $\|\cdot\|$ on $\R^d$ so that $\| e_j\| =1$ for all unit vectors $e_j$.
Furthermore, we shall use the notation $\cale:=\overline\R_+^d\setminus{\{\bnull\}}$
 with $\overline{\R}_+=[0,\infty]$, $\bnull$ is the $d-$dimensional vector with entries all equal to 0, and  $\calb=\calb(\mathcal{E})$ denotes the Borel $\si$-algebra with respect to the so-called {one point uncompactification}.

\bde\label{RV}
A random vector $X$ with state space $\mathcal{E}$
is called {\em multivariate regularly varying} if there is a Radon measure   $\mu\not\equiv 0$ on $\mathcal{B}(\cale)$ with $\mu(\overline\R_+^d\setminus\R^d_+)=0$ and
\begin{gather}\label{basrakmu}
\frac{\P(X\in t\cdot)}{\P(\|X\|> t)}\stv \mu(\cdot),\quad\tto,
\end{gather}
where $\stv$ denotes vague convergence.
In this case there exists some $\al>0$ such that the limit measure is homogeneous of order $-\al$:
$$\mu(u S)=u^{-\al}\mu(S),\quad u>0,$$
for every  $S\in\calb(\mathcal{E})$ satisfying $\mu(\partial S)=0$. The measure $\mu$ is called \textit{intensity measure of $X.$}
\\[2mm]
The {\em  tail index} $\al > 0$ is also called the {\em index of regular variation} of $X$, and we write $X \in \mathcal{R}(-\alpha).$
\ede

Regular variation of $V$ implies regular variation of $F$ under a Breiman condition on the weight matrix $W$ from \eqref{A}.
We shall use the following result, which is based on Proposition A.1 in \cite{Basrak200295}.

\begin{proposition}\label{th1}
Let $V:=(V_1,\ldots,V_d)^\top$ be multivariate regularly varying having components with Pareto-tails $\P (V_j>t) \sim K_j t^{-\al}$ as $\tto$ for $K_j,\al>0$ as in \eqref{pareto}
with  intensity measure $\mu$ as in \eqref{basrakmu}.
Furthermore, let the weight matrix $W:\Omega\rightarrow \R_{+}^{q\times d}$ satisfy $\E[ \|W\|^{\alpha+\delta}] < \infty$ for some $\delta>0$.
Then  the random vector $F=AV$ with $A$ as in \eqref{A}  belongs to $\calr(-\al)$.
{Let $h:\overline{\R}^q\setminus \{0\}\to \overline{\R}^k\setminus \{0\}$ for $k\in\N$ be
a continuous 1-homogeneous function.}
Then we have on $\mathcal{B}(h(\overline{\R}_+^{q}\setminus\{\bnull\}))$:
 \beam\label{hlimit}
 \frac{\pr{h(F)\in t \cdot }}{\pr{\|V\|>t}} \stv \E\mu\{x\in\R_+^d : h(A x)\in \cdot\},\  t \rightarrow \infty.
 \eeam
\end{proposition}

\begin{proof}
{Vague convergence of $F$ is given by} Proposition~A.1 in \cite{Basrak200295} 
and is equivalent to
\begin{gather}
 \frac{\pr{F\in t B }}{\pr{\|V\|>t}} \rightarrow \E\mu\{x\in\R_+^d : A x\in B\}, t \to \infty,
\end{gather}
for all relatively compact sets $B\in \mathcal{B}(\overline{\mathbb{R}}^{q}_+\setminus \{0\})$ with $\E\mu \circ A^{-1}(\partial B)=0 $.
Furthermore, by 1-homogeneity of $h$, for $t > 0$,  $\{ h(F) \in tB \}=\{ F \in t h^{-1}(B) \}.$ 
Note also that every $B$, which  is bounded away from zero,
is relatively compact in the topology we use.
Since $h^{-1}(B)$ is bounded away from zero by continuity of $h$ and the fact that $h(0)=0$, $h^{-1}(B)$ is also relatively compact. 
Moreover, $\E \mu \circ A^{-1} (\partial h^{-1} (B)) \leq \E \mu \circ A^{-1} \circ h^{-1}(\partial B )$.\\
Putting all this together, for  every relatively compact set $B \in \mathcal{B}(h(\overline{\R}_+^{q}\setminus\{\bnull\}))$ with $\E \mu \circ A^{-1} \circ h^{-1}(\partial B )=0$
we have, as $t \rightarrow \infty$, 
\beao\label{ohnehlimit}
 \frac{\pr{h(F)\in tB  }}{\pr{\|V\|>t}} = \frac{\pr{F\in th^{-1}(B)  }}{\pr{\|V\|>t}}  \rightarrow \E\mu\{x\in\R_+^d : A x\in h^{-1}(B)  \}=\E\mu\{x\in\R_+^d : h(A x)\in (B)  \}.
 \eeao
This is equivalent to vague convergence in \eqref{hlimit}.
\end{proof}

For the  extreme dependence case corresponding to vectors $V$ with asymptotically independent components the reference measure $\mu$ has support on the axes; in the extreme dependence case corresponding to vectors $V$ with asymptotically fully dependent components the reference measure $\mu$ has support  on the line $\{sK^{1/\al}1 : s>0\}$. This difference in support is reflected in the difference between \eqref{VaRconst}
and \eqref{VaRconst_dep} and affects the behaviour of aggregated exposures.


\bpr\label{singleFasym}
Assume the situation of Proposition~\ref{th1}. For the aggregated exposures  $h(F)$ we obtain
\beao \label{agrregatedtail}
\P(h(F)> t)\sim   C^h  t^{-\alpha}, \quad\tto,
\eeao
with
\begin{gather}\label{ChindChdep}
C^h=C^h_{\ind}=\sum_{j=1}^d K_j \E h^\alpha( A e_j )\quad\mbox{and}\quad
C^h =C^h_\dep(h)=\E h^\alpha(AK^{1/\al}1)
\end{gather} 
if $V_1,\dots, V_d$ are asymptotically independent or asymptotically fully dependent, respectively.
\epr

\begin{proof}
 The assertion can be shown in an analogous way to Theorem~3.4 of \cite{KKR1}.
\end{proof}

The following result gives limit relations in the most general situation, without any restriction on the dependence in the exposure vector.

\bthe\label{rvagg}
Let $g,h:\R_+^q\rightarrow \R_+$ be continuous $1-$homogeneous functions and assume the situation of Proposition~\ref{th1}.
Then for $u\in (0,\infty)$, the following assertions hold:
\begin{itemize} 
\item[(a)]
\, $\lim\limits_{\tto}\pr{g(F) >t\mid h(F)>ut }
\, = \, u^{\al} \dfrac{\E \mu\circ A^{-1} (\{x\in\R_+^q : h(x) >u, g(x)>1\})}
{\E \mu\circ A^{-1} (\{x\in\R_+^q : h(x) >1\})}.$
\item[(b)] If $g$ is additionally  bounded or has compact support on $\overline{\R}_+^{d}\setminus \{ 0\},$ then
\beam\label{ESgh}
\lim_{\tto}\E[g(F) \mid h(F)>t ]
& \sim &  \frac{t}{\tilde\mu(\{ h(x) >1\})} \int_{h(x)>1} g(x) \tilde\mu(dx),
\eeam
where {$\tilde\mu(\cdot) = \E \mu\circ A^{-1} (\cdot)$.}
\end{itemize} 
\ethe

\bproof
(a) We use Proposition~\ref{th1} to obtain
\beao
\pr{g(F) >t\mid h(F)>ut }
& = &\int_{h(F)>ut \atop g(F)>t} \frac{d\P}{\pr{h(F)>ut}}\\
& = & \int_{h(x)>u \atop g(x)>1} \frac{\pr{F\in t dx}}{\pr{\|V\|>t}} \frac{\pr{\|V\|>t}}{\pr{h(F)>ut}}.
\eeao
The second ratio converges by Proposition~\ref{th1}(a) and also the first, when taking there for $h$ the identity function.  The result follows then by vague convergence.\\
(b) Using 1-homogeneity of $g$ and Proposition~\ref{th1},
\beam
\E[g(F)\mid h(F) > t] &=& \frac1{\P(h(F)>t)}\int_{h(x)>t} g(x) \P(F\in dx) \nonumber \\
&=& \frac1{\P(h(F)>t)}\int_{h(x)>1} g(tx) \P(F\in t dx)\nonumber \\
&=& \frac{\P(\|V\|>t)}{\P(h(F)>t)}\int_{h(x)>1} {g(tx)}  \frac{\P(F\in t dx)}{\P(\|V\|>t)} \nonumber \\
& \sim & \frac{t}{\tilde\mu(\{ h(x) >1\})} \int_{h(x)>1} g(x) \tilde\mu(dx). \label{weakorvague}
\eeam
Recall that the sequence of bounded measures in \eqref{hlimit} converges to a bounded measure vaguely if and only if it converges weakly to this measure, see Theorem~2.1.4 in \cite{BasrakPhD} for further details. Hence, either assumptions  on $g$ in  (b) 
is sufficient to achieve convergence for \eqref{weakorvague}.
\eproof

\bco\label{rvaggind}
Let $u\in(0,\infty)$ and assume the situation of Proposition~\ref{th1}.
Recall the constants $C_{\ind}^h $ and $C_\dep^h$ from \eqref{ChindChdep}.\\
(a) \, If $V_1,\dots, V_d$ are asymptotically independent, then
\begin{gather} \label{tailindepgh}
   \lim_{t\to \infty}  \pr{ g(F)>t \mid h(F)>ut}= (C^h_{\ind})^{-1}\sum_{j=1}^{d}  \E \min\{ h^{\alpha}(AK^{1/\al}e_j), u^{\alpha}g^{\alpha}(A K^{1/\al}e_{j}) \}.
    \end{gather}
(b) \, If $V_1,\dots, V_{d}$ are asymptotically fully dependent, then
\begin{gather} \label{taildepgh}
   \lim_{t\to \infty}  \pr{ g(F)>t \mid h(F)>ut}=  ( C^h_\dep )^{-1}\E \min\{ h^{\alpha}(AK^{1/\al}\bone), u^{\alpha}g^{\alpha}(AK^{1/\al}1) \}
    \end{gather}
\eco

\bproof
The proof is similar to the proof of Theorem~3.4 of \cite{KKR1}.
For asymptotically independent claims $V_1,\dots,V_d$ we obtain {by Theorem~\ref{rvagg}(a) for the numerator}
\begin{eqnarray*}
\E \mu \circ A^{-1}(\{ h(x)>u, g(x)>1 \})
&=&(\sum_{j=1}^d K_j)^{-1}\sum_{j=1}^{d} K_j\E \min\{u^{-\al} h^{\alpha}(Ae_j), g^{\alpha}(Ae_{j})  \},
\end{eqnarray*}
and the expression in the denominator is {
\beam\label{nom}
(\sum_{j=1}^d K_j)^{-1}\sum_{j=1}^{d} K_j \E\{u^{-\al} h^{\alpha}(Ae_j)\} = (\sum_{j=1}^d K_j)^{-1} C_\ind^h, 
\eeam }
which yields \eqref{tailindepgh}.
In the case of asymptotically fully dependent claims we get by Theorem~3.4(b),
\begin{gather*}
\E \mu \circ A^{-1}(\{ h(x)>u, g(x)>1 \})= \|K^{1/\al}1\|^{-\al} \E \min\{u^{-\al} h^{\alpha}(AK^{1/\al}1), g^{\alpha}(AK^{1/\al}1)\},
\end{gather*}
giving with corresponding nominator relation \eqref{taildepgh}.
\eproof

\brem
In extreme value theory the {\em tail dependence coefficient} is usually defined for two possibly  dependent random variables $X_1,X_2$ with the same marginal distribution function as
$\lim_{x\to\infty} \P(X_2>x\mid X_1>x)$ provided that this limit exists (e.g. Section~9.5  in \cite{Beirlant}).
The resulting number is interpreted as a measure describing coinciding large losses.
The conditional probabilities in Definition~\ref{def:systemicriskmeasures} (a)-(c) are defined via such conditional probabilities, allowing for asymmetry.
As a consequence of regular variation the limits of the following conditional probabilities  as well as the conditional expectations can be computed explicitly: if 
 $\gamma_g/\gamma_h \rightarrow 1 $, then 
\beao
\lefteqn{\pr{g(F)>\VaR_{1-\ga_g}(g(F))\mid h(F)> \VaR_{1-\ga_h}(h(F))}}\\
&\sim \pr{ h(F)> \VaR_{1-\ga_h}(h(F))\mid   g(F)>\VaR_{1-\ga_g}(g(F))}
\eeao
and if $\gamma=\gamma_g=\gamma_h$, then we recognize 
\begin{gather}\label{coeffoftaildep}
\lim_{\gamma \to 0} \pr{g(F)>\VaR_{1-\ga}(g(F))\mid h(F)> \VaR_{1-\ga} }
\end{gather}
as the usual (symmetric) {tail dependence coefficient}, e.g. defined in \cite{Beirlant}, p.~343, eq. (9.75).
\erem

\bco\label{rvaggindES}
Let $u\in(0,\infty)$ and assume the situation of Proposition~\ref{th1}.
Recall the constants from \eqref{ChindChdep}. \\
(a) \, If $V_1,\dots, V_d$ are asymptotically independent, we find
\begin{gather} \label{tailindepESgh}
   \E[ g(F) \mid h(F)>t ] \sim \frac{\al}{\al-1} (C^h_{\ind})^{-1}\sum_{j=1}^{d}
   \E g(AK^{1/\al}e_j) h^{\alpha-1}(A K^{1/\al} e_{j}) t.
    \end{gather}
(b) \, If $V_1,\dots, V_{d}$ are asymptotically fully dependent, we find
\begin{gather} \label{taildepESgh}
   \E [g(F)\mid h(F)>t] \sim  \frac{\al}{\al-1} ( C^h_\dep )^{-1}
   \E g(AK^{1/\al}1) h^{\alpha-1}(AK^{1/\al}1) t.
    \end{gather}
(c) \, For $g=h$ we obtain the classical Conditional Tail Expectation \eqref{CoT}.
\eco

\bproof
(a) We  evaluate the integral in \eqref{ESgh} as
\begin{align*}
& \E\int_{h(x)>1}g(x) \mu\circ A^{-1}(dx)\\
&= (\sum_{j=1}^d K_j)^{-1} \E \sum_{j=1}^{d}\int_{h(x)>1, x\in \{u AK^{1/\alpha}e_j:u>0 \}} g(x) \nu^{*}(\{ se_{j} \in \R^{d} :  s A K^{1/\al}e_{j} \in dx \}  ),
\end{align*}
where  the measure  $\nu^{*}$, called the canonical exponent measure, is 
 related to the exponent measure $\nu$ of the vector $V=(V_{1},\dots,V_{d})$ by $\nu=\nu^{*}\circ K^{-1/\al}$, see Lemma~2.2 in \cite{KKR1}.
For independent components $\nu^{*}$ is concentrated on the axes.
We take into account that, whenever $x \in \{ u AK^{1/\al}e_{j} :\ u >0  \}$, the equality
\begin{gather}
\nu^{*}(\{ se_{j} \in \R^{d} :  s A K^{1/\al}e_{j} \in dx \}  ) = \alpha u^{-\alpha-1}du
\end{gather}
 holds.
Integration over the set $\{ u>1/ h(AK^{1/\al}e_j) \}$ yields
\begin{gather*}
\int_{1/h(AK^{1/\al}e_j)}^\infty\alpha  g(AK^{1/\al}e_j)u^{-\alpha} du = \frac{\alpha}{\alpha-1} g(AK^{1/\al}e_{j}) h^{\alpha-1}(AK^{1/\al}e_j),
\end{gather*}
implying
\begin{gather}
\int_{h(x)>1}g(x)\E \mu\circ A^{-1}(dx)= \frac{\alpha}{\alpha-1}(\sum_{j=1}^d K_j)^{-1}  \sum_{j=1}^{d}  \E[g(AK^{1/\al}e_{j})  h^{\alpha-1}(AK^{1/\al}e_j)].
\end{gather}
{Since $\tilde\mu(\{h(x)>1\}) = \int_{h(x)>1} \E\mu\circ A^{-1}(dx) = (\sum_{j=1}^d K_j)^{-1} C_\ind^h$}, we get \eqref{tailindepESgh}. \\[2mm]
(b) To show \eqref{taildepESgh}, 
recall that in the presence of full dependence the canonical exponent measure $\nu^*$ for fully dependent components is concentrated on the diagonal $\{ u {1}\in \R^{d}: u >0  \}$ and connected to the exponent measure $\nu$ of $V$ by $\nu= \nu^{*}\circ K^{-1}$, see also Lemma~4.2 in \cite{KKR1}. Hence,
\begin{align*}
&\int_{h(x)>1}g(x) \E \mu\circ A^{-1}(dx)\\
&=\| K^{1/\al}{1}\|^{-\al} \E \int_{h(x)>1, x \in \{uA K^{1/\al}{1}:u>0 \}}g(x)   \nu^{*}(\{ s {1} \in \R^{d} :  s A K^{1/\al}{1} \in dx \}  ).
\end{align*}
For  $x\in \{ u AK^{1/\al}{1}:\ u>0  \}, $ we have $\E \nu(\{ sK^{1/\al}{1} \in \R^{d} :  s AK^{1/\al}{1} \in dx \}  )= \alpha u^{-\alpha-1}du,$
which yields
\begin{align*}
\int_{h(x)>1}g(x)\E \mu\circ A^{-1}(dx)&=
\|K^{1/\al}{1}\|^{-\alpha} \E \int_{1/h(AK^{1/\al}{1})}^{\infty}\alpha u^{-\alpha}g(AK^{1/\al}{1})du\\
&=\|K^{1/\al}{1}\|^{-\alpha} \frac{\alpha}{\alpha-1}\E h^{\alpha-1}(AK^{1/\al}{1}) g(AK^{1/\al}{1}).
\end{align*}
This leads to \eqref{taildepESgh}.
\eproof


\section{The conditional systemic  risk measures}\label{s32}

We are now ready to  investigate  the conditional systemic risk measures from Definition~\ref{def:systemicriskmeasures} of a financial or insurance market
based on the bipartite graph represented by the random matrix $A=(A_{ij})_{i,j=1}^{q,d}$ as in \eqref{A} with $q$ agents and $d$ objects.

First, we  assess to which extent the risk of agent $i$ is affected by high market losses.
Second, we evaluate the influence of the individual agent's risk to the market risk, reflecting the systemic importance of an individual agent.
Third, we consider the influence of the risk of agent $k$ on the risk of agent $i$. 
Throughout this section we assume that the claims $V_1, \ldots, V_d$ are asymptotically independent.

In this section we return  to the multivariate risk measures from Definition ~\ref{def:systemicriskmeasures}  applied to aggregation functions;  we take again  $g(F)$ as the projection on some  component and $h(F)=\|F\|$ as a norm; here we can even allow for $\| \cdot \|$ to be only a quasi-norm. 
In particular this norm, or quasi-norm, does not have to equal the  reference norm in the definition of regular variation in \eqref{basrakmu}. 

The following result determines the probability of joint large losses for individual institutions and the financial system in different conditional situations.

\begin{proposition} \label{VaRcond}
Let $V_1,\dots,V_d$ be asymptotically independent and $u>0$. Assume that the conditions of Proposition~\ref{th1} are satisfied. Moreover 
assume that $\ga\to 0$ and $ \kappa\in (0,\infty)$. 
Then
\beam
\label{eq:VaR:FigivenFindep}
\pr{F_i>\VaR_{1-\ga \kappa}(F_i)\mid \|F\|> \VaR_{1-\ga}(\|F\|)}
& \to &
\sum_{j=1}^{d} K_j  \E \min\Big\{  \frac{\|Ae_j\|^\al}{C^S_\ind}, \kappa \frac{ A_{ij}^{\alpha}}{C^{i}_\ind}\Big\}\\
\label{eq:VaR:FgivenFiindep} \pr{\|F\|> \,\VaR_{1- \kappa \ga}(\|F\|)\mid F_i>\VaR_{1-\ga }(F_i)}
& \to &
\sum_{j=1}^{d} K_j \E \min\Big\{ \kappa \frac{\|Ae_j\|^\al}{C^S_\ind},   \frac{A_{ij}^{\alpha}}{C^{i}_\ind} \Big\}\\
\label{eq:VaR:FigivenFkindep}
\pr{F_i> \,\VaR_{1-\ga \kappa}(F_i)\mid F_k>\VaR_{1-\ga }(F_k)} & \to &
\sum_{j=1}^{d} K_j \E \min\Big\{ \kappa \frac{A_{ij}^\al}{C^i_\ind}, \frac{A_{kj}^{\alpha}}{C^k_\ind} \Big\}.
\eeam
Moreover, for the Conditional Tail Expectations, if $\alpha > 1  $ then 
\beam\label{ExpectedShortfallFi}
\ICoTE_{1-\gamma}(F_{i}  \mid  \| F \|) &\sim &\frac{\alpha}{\alpha-1}(C^S_\ind)^{1/\al-1}
\sum_{j=1}^d K_j \E[A_{ij}\|Ae_j\|^{\al-1}] \gamma^{-1/\al}.\\
\SCoTE_{1-\gamma}(  \| F \| \mid F_{i}) &\sim &\frac{\al}{\al-1}(C^{{i}}_\ind)^{1/\al-1}\sum_{j=1}^{d}K_j \E[ A_{ij}^{\alpha-1}\|Ae_j\|]\gamma^{-1/\al}.\label{ExpectedShortfallMulti}\\
\MCoTE_{1-\gamma} (F_i\mid F_k) & \sim & \frac{\al}{\al-1}(C^k_\ind)^{1/\al-1}\sum_{j=1}^{d}K_j\E[ A_{kj}^{\alpha-1} A_{ij}]\gamma^{-1/\al}.
\label{ExpectedShortfallMutual}
\eeam
\end{proposition}

\bproof We show the following, slightly more general result: Let 
$$g, h \in \{f:\R_+^q\rightarrow \R_+\,; \, f({{x}}) = \| {{x}}\|\mbox{ and } f_k: \R_+^q\rightarrow \R_+; f_k({{x}}) = x_k, \, k=1, \ldots, q\},$$
 then under the assumptions of this proposition, 
\beam\label{eq:VaR:fgind}
\lefteqn{\pr{g(F)>\VaR_{1-\ga \kappa}(g(F))\mid h(F)> \VaR_{1-\ga}(h(F))} }\quad\quad\quad\nonumber\\
& \to &
\sum_{j=1}^{d} \E \min\Big\{ \frac{h^\al(AK^{1/\al}e_j)}{C^h_{\ind}}, \frac{ \kappa g^\al(AK^{1/\al}e_j)}{C^g_{\ind}} \Big\}\label{eq:VaR:FgivenFiindep}, \ \gamma \to 0.
\eeam
To show this general result we  set $\VaR_{1-\ga \kappa }(g(F))=t$ and $\VaR_{1-\ga}(h(F))=ut$.
Now recall that by Lemma \ref{Cor:VaR}
$$
\VaR_{1-\gamma}(F_{i})\sim (C^{i}_{\ind})^{1/\alpha} \gamma^{-1/\alpha} 
\quad\mbox{ and  } \quad
\VaR_{1-\gamma}(\|F\|) \sim ( C^S_\ind)^{1/ \alpha}\gamma^{-1/\alpha},\ \gamma \rightarrow 0.
$$ 
This implies that
$$u = \frac{\VaR_{1-\ga}(h(F))}{\VaR_{1-\ga \kappa}(g(F))} = \frac{ (C_{\ind}^h)^{1/\al} \ga^{-1/\al}}{(C_{\ind}^g)^{1/\al} (\ga \kappa)^{-1/\al}} (1+o(1)),\ \gamma\to 0$$
such that
$$u^\al = \frac{ C_{\ind}^h }{C_{\ind}^g }\frac{\ga \kappa}{\ga } (1+o(1)) = \frac{ C_{\ind}^h }{C_{\ind}^g } \kappa (1+o(1)),\ \gamma \to 0.$$
We conclude that \eqref{eq:VaR:fgind}  holds  by Corollary~\ref{rvaggind}.

The analogous expressions for the  Conditional Tail Expectation follow immediately from {Corollary}~\ref{rvaggindES}.
\eproof

\brem
Here is  an interpretation of \eqref{eq:VaR:FigivenFindep} from the viewpoint of a regulator. Assume that  the market situation changes from its normal situation such that, for instance, $\|F\|> \delta\VaR_{1-\ga}(\|F\|)$  for some $\delta>1$, then by \eqref{tailindepgh} the factor $\kappa$ in the limit in \eqref{eq:VaR:FigivenFindep} becomes $\delta^\alpha\kappa$. 
For the sake of argument we call the first contribution of the sum on rhe right-hand side of  ~\eqref{eq:VaR:FigivenFindep}, $  \frac{\|Ae_j\|^\al}{C^S_\ind}$   the systemic constant and the second contribution, $ \kappa \frac{ A_{ij}^{\alpha}}{C^{i}_\ind}$,  the individual contribution.
Consider the situation, where the minimum has been attained by the individual contribution  under  previous market conditions. 
Then the change in the market, which resulted in the change of $\kappa$ to $\delta^\alpha\kappa$ for some $\delta>1$ can result in two situations. In the first one, the minimum is still assumed by the individual contribution, even though it is increased by the factor $\delta^\al$.
If the systemic change is so substantial that the individual contribution becomes larger than the systemic constant, then the limit on the right-hand side of \eqref{eq:VaR:FigivenFindep} becomes the systemic constant.

If a regulator implements the strategy that the limiting conditional probability remains the same under all market conditions, then it would firstly require that under normal market conditions
$$\E\min\Big\{  \frac{\|Ae_j\|^\al}{C^S_\ind}, \kappa \frac{ A_{ij}^{\alpha}}{C^{i}_\ind}\Big\} = \kappa \frac{ \E[A_{ij}^{\alpha}]}{C^{i}_\ind}.$$
If the market comes under stress, then the regulator would
raise the Value-at-Risk for the institution as long as the minimum is still taken by the individual contribution.
The situation, however, that the minimum is taken by the systemic constant indicates that the stressed market condition can no longer be absorbed by an adjustment of the individual capital reserves. Then political measures have to be taken.

The other limit relations have analogous interpretations.
\erem

For each of the limiting expressions in Proposition \ref{VaRcond} the limiting behaviour for $\kappa \rightarrow 0$ is linear, as is made precise in the next proposition. 

We assume that there exist constants $w, W > 0  $ such that  $0 < w \le W_{ij} \le W  $ and that there  exist constants $b, B$ such that $0 < b \le \|Ae_j\| \le B$.
For example, if $W_{ij} = \deg(j)^{-1} \bone (i \sim j)$  then we can take $w = \frac{1}{1 + d}$ and $W=1$. We set 
\beam \label{kappai} 
\kappa_0  =   \kappa_0(i)  = \frac{b^\al}{C^S_\ind} \,  \frac{ C^{i}_\ind}{W^{\alpha}} ,\quad
\kappa_1  =   \kappa_1(i) = \frac{ C^{S}_\ind} {C^i_\ind}  \, \frac{b^\alpha}{W^{\alpha}},
\quad\mbox{and}\quad
\kappa_2 = \kappa_2(i,k) = \frac{ C^i_\ind }{ C^k_\ind} \,\frac{w^\alpha}{W^\alpha }.
\eeam 
Moreover,  we define
\beam \label{taui}
\tau(i) =  \sum_{j=1}^{d} K_j \E  \bone (i \sim j) \frac{\|Ae_j\|^\al}{C^S_\ind}
\quad\mbox{ and }\quad
\tau(i,k) = \sum_{j=1}^{d} K_j \E  \bone (k \sim j) \frac{A_{ij}^\al}{C^i_\ind};
\eeam 
and note that $\tau(i) \le 1$ and $\tau(i,k) \le 1$ 
through the definitions of $C^S_\ind$ and $C^i_\ind$, respectively. 
If $i$ and $k$ do not share an object then $\tau(i,k)=0$. 

 \begin{theorem}\label{help_covar}
 Assume that the conditions for Proposition \ref{VaRcond} hold and that there exist finite constanst $w, W > 0  $ such that  $0 < w \le W_{ij} \le W$ and also finite constants $b, B$ such that $0 < b \le \|Ae_j\| \le B$. 
 \begin{itemize}
 \item[(a)] For $\kappa \le \kappa_0$, 
\beam \label{icovarbasis}
\sum_{j=1}^{d} K_j \E \min\Big\{  \frac{\|Ae_j\|^\al}{C^S_\ind}, \kappa \frac{ A_{ij}^{\alpha}}{C^{i}_\ind}\Big\} = \kappa \sum_{j=1}^{d} K_j \E \left\{ \frac{ A_{ij}^{\alpha}}{C^{i}_\ind} \right\} = \kappa. 
\eeam
\item[(b)] 
For $\kappa\le  \kappa_1(i) $, 
\beam \label{scovarbasis} 
\sum_{j=1}^{d} K_j \E \Big\{\min\Big\{\kappa \frac{\|Ae_j\|^\al}{C^S_\ind},   \frac{A_{ij}^{\alpha}}{C^{i}_\ind} \Big\} =  \kappa \tau(i).
\eeam
\item[(c)]  For $\kappa \le \kappa_2 (i,k) $, 
\beam \label{mcovarbasis} 
\sum_{j=1}^{d} K_j \E \min\Big\{ \kappa \frac{A_{ij}^\al}{C^i_\ind}, \frac{A_{kj}^{\alpha}}{C^k_\ind} \Big\} = 
\kappa \tau(i,k).
\eeam
 \end{itemize} 
 \end{theorem}

\bproof 
 To show \eqref{icovarbasis} we start with \eqref{eq:VaR:FigivenFindep}. Consider the expression 
$$ \min\Big\{  \frac{\|Ae_j\|^\al}{C^S_\ind}, \kappa \frac{ A_{ij}^{\alpha}}{C^{i}_\ind}\Big\} = \bone (i \sim j)  \min\Big\{  \frac{\|Ae_j\|^\al}{C^S_\ind}, \kappa \frac{ W_{ij}^{\alpha}}{C^{i}_\ind}\Big\}.
$$ 
If $i \not\sim j$ then the minimum is 0, and if $i \sim j$ then we can choose 
\beam \label{kappabound} \kappa < \frac{\|Ae_j\|^\al}{C^S_\ind} \,  \frac{ C^{i}_\ind}{ W_{ij}^{\alpha}}.\eeam
While this expression is random,  $\kappa_0$ is not random, and 
for $\kappa \le \kappa_0, $ 
\eqref{kappabound} is satisfied for any realisation of the network. Hence 
\beao
\E \min\Big\{  \frac{\|Ae_j\|^\al}{C^S_\ind}, \kappa \frac{ A_{ij}^{\alpha}}{C^{i}_\ind}\Big\} =  \kappa \E \Big\{ \bone (i \sim j)  \frac{ W_{ij}^{\alpha}}{C^{i}_\ind}\Big\} .
\eeao
Summing over $j=1,\ldots, d$ and recalling the definition of $C^{i}_\ind$ gives \eqref{icovarbasis}.
 
For \eqref{scovarbasis} we start with 
\eqref{eq:VaR:FgivenFiindep}; the argument is similarly straightforward. 
Consider the expression 
$$ 
\min\Big\{{\kappa} \frac{\|Ae_j\|^\al}{C^S_\ind},   \frac{A_{ij}^{\alpha}}{C^{i}_\ind} \Big\}
= \bone (i \sim j) \min\Big\{\kappa \frac{\|Ae_j\|^\al}{C^S_\ind},   \frac{W_{ij}^{\alpha}}{C^{i}_\ind} \Big\}.
$$
If 
$$\kappa  \le  \frac{W_{ij}^{\alpha}}{C^{i}_\ind} \,  \frac{C^S_\ind}{\|Ae_j\|^\al},$$ 
then 
$$ \bone (i \sim j) \min\Big\{\kappa \frac{\|Ae_j\|^\al}{C^S_\ind},   \frac{W_{ij}^{\alpha}}{C^{i}_\ind} \Big\} = \bone (i \sim j) \kappa \frac{\|Ae_j\|^\al}{C^S_\ind}.$$
In particular, this 	equation holds for $\kappa \le \kappa_1$  with $\kappa_1$ given in \eqref{kappai}. Again summing over all $j$ gives \eqref{scovarbasis}.

 To show \eqref{mcovarbasis} we use 
 \eqref{eq:VaR:FigivenFkindep}
Consider the expression 
$$ 
 \min\Big\{ \kappa \frac{A_{ij}^\al}{C^i_\ind}, \frac{A_{kj}^{\alpha}}{C^k_\ind} \Big\} = \bone (i \sim j) \bone (k \sim j) \min\Big\{ \kappa \frac{W_{ij}^\al}{C^i_\ind}, \frac{W_{kj}^{\alpha}}{C^k_\ind} \Big\}.
$$ 
For $\kappa \le \kappa(i,k), $  
$$ \bone (i \sim j) \bone (k \sim j) \min\Big\{ \kappa \frac{W_{ij}^\al}{C^i_\ind}, \frac{W_{kj}^{\alpha}}{C^k_\ind} \Big\} 
= \bone (k \sim j) \kappa \frac{W_{ij}^\al}{C^i_\ind}
=  \kappa \bone (k \sim j) \frac{A_{ij}^\al}{C^i_\ind}
$$ 
for $\alpha > 0$. Summing over $j$ gives the assertion \eqref{mcovarbasis}. 
\eproof 

Following on from \eqref{icovarbasis}, \eqref{scovarbasis} and \eqref{mcovarbasis} we can now assess the limiting behaviour of ICoVaR, SCoVaR and MCoVaR from Definition \ref{def:systemicriskmeasures}, specified for the aggregation function $h(F) = \| F\| $, where   $F=(F_1,\ldots,F_q)$ is the  random exposure vector.
For $\ga_i,\ga\in (0,1)$ referring to agent $i$ and the market, respectively, we consider the following conditional systemic risk measures: \\[2mm]
(a) \,  {\em Individual Conditional Value-at-Risk}\\
\centerline{$\ICoVaR_{1-\gamma_i, \gamma}(F_i \mid  \| F\|):=\inf\{ t \geq 0: \pr{ F_i>t \mid  \| F\|> \VaR_{1-\gamma}( \| F\|) } \leq \gamma_i \},$}\\
(b) \,  {\em Systemic Conditional Value-at-Risk} \\
\centerline{$\SCoVaR_{1-\gamma,\gamma_i}( \| F\|\mid { F_i} ):=\inf\{ t \geq 0: \pr{  \| F\|>t \mid F_{i} > \VaR_{1-\gamma}(F_i)  } \leq \gamma \},$}\\
(c) \,  {\em Mutual Conditional Value-at-Risk}\\
\centerline{$\MCoVaR_{1-\gamma_i, \gamma_k}(F_i \mid F_k):=\inf\{ t \geq 0: \pr{ F_i>t \mid F_{k} > \VaR_{1-\gamma_k}(F_k)} \leq \gamma_i \}$.}

\begin{theorem} \label{asymprisk} 
Assume that there exist constants $w, W > 0  $ such that  $0 < w \le W_{ij} \le W  $ and that there is an upper bound $a<\infty$ such that $\|Ae_j\| \le a$.
Recall the constants from \eqref{kappai} and \eqref{taui}
\begin{enumerate}
\item As $\gamma \rightarrow 0,$ for $\gamma_i \le \kappa_0 (i) $, 
\beam \label{icovarasymp} \ICoVaR_{1-\gamma_i, \gamma}(F_i \mid  \| F\|) \sim \VaR_{1 - \gamma_i \gamma}(F_i) \sim (C^i_\ind)^\frac{1}{\alpha} (\gamma_i \gamma) ^{-\frac{1}{\alpha}}; 
\eeam 
\item  As $\gamma_i \rightarrow 0, $ for 
$\gamma \le \kappa_1(i) \tau(i)$, 
\beam  \label{scovarasymp}  \SCoVaR_{1-\gamma,\gamma_i}( \| F\|\mid { F_i} ) \sim VaR_{1 - \frac{\ga_i \ga }{\tau(i) }} ( \| F\|) \sim (C^S_\ind)^\frac{1}{\alpha} \left\{ \frac{\ga_i \ga }{\tau(i) }\right\}  ^{-\frac{1}{\alpha}};
\eeam 
\item If $\tau (i,k) \ne 0$, then as $\gamma_k \rightarrow 0,$  for $\gamma_i \le \kappa_2(i,k) \tau(i,k)$, we have 
\beam   \label{mcovarasymp} \MCoVaR_{1-\gamma_i, \gamma_k}(F_i \mid F_k) \sim VaR_{1 - \frac{\gamma_i \gamma_k }{\tau(i,k)} } ( F_i) \sim  (C^i_\ind)^\frac{1}{\alpha} \left\{ \frac{\gamma_i \gamma_k }{\tau(i,k)}\right\}  ^{-\frac{1}{\alpha}};
\eeam 
and if $\tau (i,k) =0$ then, as $\gamma_i \rightarrow 0$,  
$$ \MCoVaR_{1-\gamma_i, \gamma_k}(F_i \mid F_k) \sim VaR_{1 - \gamma_i} (F_i) \sim  (C^i_\ind)^\frac{1}{\alpha} \gamma_i^{-\frac{1}{\alpha}}. $$ 
\end{enumerate} 
\end{theorem}

\begin{proof}
First, from \eqref{eq:VaR:FigivenFindep} and \eqref{icovarbasis}, for $\kappa \le \kappa_0 = \kappa_0(i) $, as $\gamma \rightarrow 0$, 
\beao
\pr{F_i>\VaR_{1-\ga \kappa}(F_i)\mid \|F\|> \VaR_{1-\ga}(\|F\|)}
&\to & 
\sum_{j=1}^{d} K_j  \E \min\Big\{  \frac{\|Ae_j\|^\al}{C^S_\ind}, \kappa \frac{A_{ij}^{\alpha}}{C^{i}_\ind}\Big\} \, = \, \kappa. 
\eeao 
Hence for $\gamma_i \le \kappa_0$, $\gamma \rightarrow 0$,
$$\pr{ F_i>\VaR_{1-\ga \gamma_i}(F_i) \mid  \| F\|> \VaR_{1-\gamma}( \| F\|) } \sim \gamma_i.
$$ 
Thus $ \ICoVaR_{1-\gamma_i, \gamma}(F_i \mid  \| F\|) \sim \VaR_{1 - \gamma_i \gamma}(F_i)$. The asymptotics for the VaR follow from Lemma~\ref{Cor:VaR}, yielding \eqref{icovarasymp}. 

For \eqref{scovarasymp}, \eqref{eq:VaR:FgivenFiindep} and \eqref{scovarbasis} give that for $\gamma \rightarrow 0$ and $\kappa > \kappa_1= \kappa_1(i) $, 
\beao \pr{\|F\|> \,\VaR_{1- \kappa \ga}(\|F\|)\mid F_i>\VaR_{1-\ga }(F_i)}
& \to &
\sum_{j=1}^{d} K_j \E \min\Big\{ \kappa \frac{\|Ae_j\|^\al}{C^S_\ind},   \frac{A_{ij}^{\alpha}}{C^{i}_\ind} \Big\} \, = \,  \kappa \tau(i).
\eeao 
In particular, a simple rescaling gives  
\beao
 \pr{\|F\|> \,\VaR_{1- \kappa {\ga_i}} (\|F\|)\mid F_i>\VaR_{1-\ga_i }(F_i)}
& \to &
\kappa \tau(i),\ \gamma_i \rightarrow 0.
\eeao 
Letting $\gamma = \kappa \tau(i)$ gives that for $\gamma \le \kappa_1 \tau(i)$
\beao
 \pr{\|F\|> \,\VaR_{1- \frac{\ga_i \ga }{\tau(i) } }(\|F\|)\mid F_i>\VaR_{1-\ga_i }(F_i)}
& \to & 
\gamma,\  \gamma_i \rightarrow 0 .
\eeao 
Now \eqref{scovarasymp} follows as before. 
For  \eqref{mcovarasymp}, \eqref{eq:VaR:FigivenFkindep} and \eqref{mcovarbasis} give that for $\gamma \rightarrow 0$ and $\kappa \le \kappa_2(i,k)$,
\beao
\pr{F_i> \,\VaR_{1-\ga \kappa}(F_i)\mid F_k>\VaR_{1-\ga }(F_k)} & \to &
\sum_{j=1}^{d} K_j \E \min\Big\{ \kappa \frac{A_{ij}^\al}{C^i_\ind}, \frac{A_{kj}^{\alpha}}{C^k_\ind} \Big\} \, = \, \kappa \tau(i,k).
\eeao 
Changing variables gives that
\beao
\pr{F_i> \,\VaR_{1-\ga_k \kappa}(F_i)\mid F_k>\VaR_{1-\ga_k }(F_k)} & \to &
\kappa \tau(i,k),\  \gamma_k \rightarrow 0.
\eeao 
Setting $\gamma = \kappa \tau(i,k)$ and requiring that $\gamma \le \kappa_2(i,k) \tau(i,k)$ gives \eqref{mcovarasymp} when $\tau(i,k) \ne 0$. 

The last assertion follows from the fact that $F_i$ and $F_k$ are independent if they do not share an object. 
\end{proof} 

\brem
The asymptotic behaviour of the risk measures is assessed in Theorem~\ref{asymprisk} through the exceedance probabilities conditioned on an extreme event. For example, in \eqref{mcovarasymp}, agent $k$ has already incurred a very large loss. This loss will have an effect on the loss of agent $i$ if they share some objects in their portfolios. The more objects they share, the larger $\tau(i,k)$ would be. The unconditional \VaR\ threshold $1 - \gamma_i$ at which $\pr {F_i > t} = \gamma_i$ has to be adjusted to  $1 - \gamma \frac{\gamma_k}{\tau(i,k)}$ if $\tau(i,k) \ne 0$. The larger $\tau(i,k)$, the larger $1 - \gamma \frac{\gamma_k}{\tau(i,k)}$, hence the more stringent the requirements on agent $i$. 

The effect of the network on the agent in \eqref{scovarasymp} indicates the dependence on $\tau(i)$, which increases with the number of connections of agent $i$. Again, the larger $\tau(i)$, the more stringent the requirements on agent $i$ should be. 

Even in \eqref{icovarasymp} there is dependence of the network structure, which is reflected in $\kappa_0(i)$ as well as in $C_{\ind}^i$. 
\erem

\section{Approximation and illustration of network effects}\label{s4}

Throughout this section we restrict ourselves to the situation  that the losses $V_1,\ldots,V_d$ are asymptotically independent.

\subsection{Losses which are not covered}\label{s41}

In the bipartite graph model, depending on the random mechanism of the agents to choose various objects, it can happen that certain objects are not chosen by any of the agents.
In the (re)insurance context of Example~\ref{largeclaims} this means that certain large losses may be not insured.
This happens for instance for certain natural catastrophes like earthquakes, where the state or the international community may be liable, see \cite{sigma,uninsured} for further information and concrete numbers.
In this subsection we approximate the probability of large losses not covered and we define
$$N = \sum_{j=1}^d \bone (\deg(j)=0) $$
as the (random) number of non-covered losses. 

\begin{proposition}\label{tax}
Assume that the  asymptotically independent objects $V_1,\ldots, V_d$ have Pareto tails given in \eqref{pareto}.
Then
\beao
\P\Big(\sum_{j=1}^d \bone (\deg(j)=0)\,  V_j>t\Big)
 \sim  t^{-\al} \,  \sum_{l=1}^d K_{l}  \,   \P(\deg(l)=0) .
\eeao
\end{proposition}

\bproof
 We condition on all possible sets $\calw$ of non-covered losses and calculate
 \beao
\lefteqn{\P\Big(\sum_{j=1}^d \bone (\deg(j)=0) V_j>t\Big)
= \sum_{w=1}^d \P(N=w) \P\Big(\sum_{j=1}^d \bone (\deg(j)=0) V_j>t\mid N=w\Big)}\\
&=&    \sum_{w=1}^d  \P(N=w)  \sum_{\calw: |\calw|=w} \P\Big(\sum_{l\in\calw} V_{l} >t\Big)  \P( \deg(j)=0, j \in \calw \mid N=w)
\\
& \sim &   t^{-\al} \, \sum_{w=1}^d   P(N=w) 
\sum_{\calw: |\calw|=w}   \, \sum_{l\in\calw} K_{l}  \P(  \deg(j)=0, j \in \calw \mid N=w) 
\eeao
as $\tto$, 
using that for the asymptotically independent regime (cf. Lemma~2.3 of \cite{KKR1},
$$  \P\Big(\sum_{l\in\calw} V_{l} >t\Big) \sim t^{- \alpha } \sum_{l\in\calw} K_{l},\quad\tto.$$
Hence, interchanging the order of summation and using  the law of total probability,
\beao
\lefteqn{\P\Big(\sum_{j=1}^d \bone (\deg(j)=0) V_j>t\Big)} \\
&\sim &  t^{-\al}  \, \sum_{w=1}^d   \P(N=w) \, \sum_{l=1}^d K_{l} 
\sum_{\calw: |\calw|=w}  \bone (l\in\calw) \,  \P( \deg(j)=0, j \in \calw \mid N=w)  \\
&=&  t^{-\al} \,  \sum_{l=1}^d K_{l}  \,   \P(\deg(l)=0).
\eeao
\eproof

\bexam
Assume that, given the number of non-insured losses is $N=w$ for $0\le w\le d$, 
all sets of these $w$ claims have the same probability to be not covered.
Then Proposition~\ref{tax} takes on a particularly simple form. In this setting, \mbox{$\P( \deg(l) = 0\mid N=w) = {w}/{d}$} for every $l\in\{1,\ldots,d\}$,
and
\beao
\P( \deg(l) = 0) &=& \sum_{w=1}^d  \P(N=w) P( \deg(l) = 0\mid N=w) 
\, = \, \frac1{d} \sum_{w=1}^d w  \,P(N=w)
\, = \, \frac1{d} \,\E N.
\eeao
Hence
$$
\P\Big(\sum_{j=1}^d \bone (\deg(j)=0) V_j>t\Big)
 \sim  t^{-\al}  \E N \frac{1}{d} \, \sum_{l=1}^d K_{l}.
 $$
Moreover, if all edges are independent and have same probability $p\in [0,1]$ to be present, then
$
\E N = \sum_{w=1}^d \P( \deg(l) = 0)
\, = \, d (1-p)^{q}.
$
In this case, the probability of large non-insured losses can be approximated by
\beao
\P\Big(\sum_{j=1}^d \bone (\deg(j)=0) V_j>t\Big) \sim t^{-\al} (1-p)^{q}\sum_{l=1}^d K_{l},\quad t\to\infty.
\eeao
\eexam

\subsection{Independent bipartite graph model: conditional systemic  risk measures}\label{s42}

{In this section we exemplify our results based on a bipartite network model, where all  edges are independent and the weighted adjacency matrix $A$ is as in Example~\ref{largeclaims},} referring to the situation of a large claims insurance market.
Hence, $A_{ij} = \frac{\bone (i\sim j)}{\deg(j)}$, where $\{\bone (i\sim j), 1\le i\le d, 1\le j\le q\}$ are independent Bernoulli random variables with $\E\bone (i\sim j)=p_{ij}$.

\subsubsection{
Poisson approximations}

If $d$ and $q$ are large, we can provide Poisson approximations for the quantities $C^S_\ind$, $C^i_\ind$, $\E A_{ij}\|Ae_j\|^{\al-1}$, $\E A_{ij}^{\alpha-1}\|Ae_j\|$, and $ \E A_{kj}^{\alpha-1} A_{ij}$ which appear in Proposition~\ref{VaRcond}. 
We define by $X\sim\Pois(\la)$ a Poisson-distributed random variable $X$  with mean $\la>0$. We shall use the following Poisson variables;
\begin{eqnarray*}
X_j^{i,k} \sim {\rm Pois}(\la_j^{i,k}) & \mbox{ with } &  \la_j^{i,k} = \sum_{l=1, l\neq i,k}^q p_{li}, \\
X_j^i\sim {\rm Pois}(\la_j^i ) & \mbox{ with } &  \la_j^i = \sum_{l=1, l\neq i}^q p_{li}, \mbox{ and}\\
X_j\sim {\rm Pois}(\la_j)& \mbox{ with } &  \la_j=\sum_{k=1}^q p_{kj}.
\end{eqnarray*} 
 Proposition 4.1 from \cite{KKR1}  gives that 
\begin{gather} \label{bij}
\left| \E A_{ij}^{\alpha} - p_{ij}\E  (1 + X_j^i)^{-\alpha}  \right|
\leq p_{ij} \min\{ 1, (\lambda_{j}^{i})^{-1} \} \sum_{k=1, \ldots, q; k \ne i} p_{kj}^{2} =: B(i,j),
\end{gather}
and, {for the $r$-norm for some $r\ge 1$,} 
\begin{gather} \label{bj}
 \left|\E \|Ae_{j}\|^{\alpha} -\E \big[ \1 \{X_j \ge 1\} (1 + X_j)^{\alpha(1/r - 1 )} \big] \right| \leq \min\{ 1, (\lambda_{j})^{-1} \} \sum_{k=1}^{q} p_{kj}^{2} =: B(j).
\end{gather}
We shall also employ 
\begin{gather} \label{bjk}
 B(i,j,k):= \min\{ 1, (\lambda_{j}^{i,k})^{-1} \} \sum_{\ell=1, \ldots, q; \ell \ne i} p_{\ell j}^{2}.
\end{gather}

The following lemma is an immediate consequence of \eqref{VaRconst}, \eqref{bij} and \eqref{bj}.

\ble \label{boundlemma} 
With the notation as above
\beam\label{bound1}
  \Big| C_{\ind}^i - \sum_{j=1}^d K_j p_{ij}\E  (1 + X_j^i)^{-\alpha} \Big|
 & \le & \sum_{j=1}^d K_j {B(i,j)}.  \\ 
 \Big| C_{\ind}^S - \sum_{j=1}^d K_j \E \big[ \1 \{X_j \ge 1\} (1 + X_j)^{-\alpha \frac{r-1}{r}} \big] \Big|
 & \le & \sum_{j=1}^d K_j B(j). \label{bound2}
 \eeam
 \ele
 
Similarly as in Lemma~\ref{boundlemma} we can derive Poisson approximations for the limiting quantities from Proposition~\ref{VaRcond}, as follows. 

\bpr \label{ex41prop} 
Assume that $A_{ij} = \frac{\bone (i\sim j)}{\deg(j)}$, where $\{\bone (i\sim j), 1\le i\le d, 1\le j\le q\}$ are independent Bernoulli random variables with $\E\bone (i\sim j)=p_{ij}$.
Define
\beam \label{M}
M_1 =  \min\left\{ \frac{\kappa }{C^{i}_\ind} ,   \frac{1}{C^S_\ind} \right\},\quad
M_2 =  \min\left\{ \frac{1}{C^{i}_\ind} ,   \frac{\kappa}{ C^S_\ind} \right\} ,\quad\mbox{and}\quad
M_3 = \min\left\{ \frac{\kappa}{C^{i}_{\ind}}, \frac{1}{C_{\ind}^{k}}     \right\} .
\eeam
Then for $r \ge 1$, for the limiting expressions of the Conditional Value-at-Risk measures, 
\beam
\Big| \E \min\Big\{  \frac{\|Ae_j\|^\al}{C^S_\ind}, \kappa \frac{ A_{ij}^{\alpha}}{C^{i}_\ind}\Big\} - p_{i,j}  \E \min\Big\{ \frac{(1 + X_j^i)^{-\alpha + \frac{\alpha}{r} }}{C^S_\ind}, \kappa \frac{  {(1 + X_j^i)^{- \alpha }}}{C^{i}_\ind}\Big\}  \Big| & \le& M_1 B(i,j) , \label{bound6}  \\
\Big| \E \min\Big\{\kappa \frac{\|Ae_j\|^\al}{C^S_\ind},\frac{ A_{ij}^{\alpha}}{C^{i}_\ind}\Big\} - p_{i,j}  \E \min\Big\{\kappa \frac{(1 + X_j^i)^{-\alpha + \frac{\alpha}{r} }}{C^S_\ind},  \frac{  {(1 + X_j^i)^{- \alpha }}}{C^{i}_\ind}\Big\} \Big| & \le& M_2  B(i,j) , \label{bound7} 
\eeam
and for $i \ne k$,
\beam
\Big| \E \min\Big\{  \kappa \frac{A_{ij}^\al}{C^i_\ind}, \frac{A_{kj}^{\alpha}}{C^k_\ind} \Big\} -   {p_{ij} p_{kj}}{M_3} \E \big[ ( 2 + X_j^{i,k} )^{-\alpha} \big] \Big|
 & \le& {p_{ij} p_{kj}}M_3 B(i,j,k)  \label{bound5} , 
\eeam
 with $B(i,j)$ given in \eqref{bij}, $B(j)$ given in \eqref{bj}, and $B(i,j,k)$ given in \eqref{bjk}.  \\
Moreover, for the limiting expressions of the Conditional Tail Expectations,  if $\alpha > 1$, then 
\beam
\Big| \E A_{ij}\|Ae_j\|^{\al-1} -  p_{ij}\E \big[ \big( 1 + X_j^i \big)^{-\frac{\alpha(r-1)+1}{r}}\big] \Big|
 & \le & B(i,j), \label{bound3}\\
 \Big| \E A_{ij}^{\alpha-1}\|Ae_j\| -  p_{ij}\E \big[\big( 1 + X_j^i\big)^{\frac{1}{r} - \alpha}\big] \Big|
 & \le & B(i,j), \label{bound4}
\eeam
and for $i \ne k$,
\beam \label{lastbound}
\left|\E  A_{kj}^{\alpha-1} A_{ij}-  p_{ij}p_{kj}\E \big[ \big( 2 + X_j^{i,k} \big)^{ \alpha}\big] \right|
 & \le &  p_{ij} p_{kj}\min\{ 1, (\lambda_{j}^{i,k})^{-1} \} \sum_{\ell=1, \ldots, q; \ell \ne i} p_{\ell j}^{2} .
\eeam
\epr

\bproof
We compute for the constants in the conditional probabilities
\beam \label{normkey}
\| A e_j\|^{\alpha }
&=&   \Big(\sum _{k=1}^q \frac{\bone (k\sim j)}{\deg(j)^r} \Big)^{\frac{\al}{r}} =
 \Big(\frac1{\deg(j)^{r-1}} \Big)^{\frac{\al}{r}} \bone(\deg(j) > 0) .
\eeam
With \eqref{normkey},
\beao
 \min\left\{  \frac{\|Ae_j\|^\al}{C^S_\ind}, \kappa \frac{ A_{ij}^{\alpha}}{C^{i}_\ind}\right\} 
 &=&  \bone (i\sim j) \min\left\{  \frac{\deg(j)^{-\alpha + \frac{\alpha}{r} }}{C^S_\ind}, \kappa \frac{  {\deg(j)^{- \alpha }}}{C^{i}_\ind}\right\}  \\
 &=& \bone (i\sim j) \min\left\{  \frac{(1 + \sum_{k \ne i} \bone (k \sim j) )^{-\alpha + \frac{\alpha}{r} }}{C^S_\ind}, \kappa \frac{  { (1 + \sum_{k \ne i} \bone (k \sim j) )^{- \alpha }}}{C^{i}_\ind}\right\}
\eeao
and consequently
\begin{gather}\label{zwischenschritt}
\E  \min\left\{  \frac{\|Ae_j\|^\al}{C^S_\ind}, \kappa \frac{ A_{ij}^{\alpha}}{C^{i}_\ind}\right\}=
p_{ij} \E  \min\left\{  \frac{(1 + \sum_{k \ne i} \bone (k \sim j) )^{-\alpha + \frac{\alpha}{r} }}{C^S_\ind}, \kappa \frac{  { (1 + \sum_{k \ne i} \bone (k \sim j) )^{- \alpha }}}{C^{i}_\ind}\right\}.
\end{gather}
Now consider the function
$k(x) = \min\big\{  \frac{(1 +x )^{-\alpha + \frac{\alpha}{r} }}{C^S_\ind}, \kappa \frac{  { (1 + x)^{- \alpha }}}{C^{i}_\ind}\big\}.$
If $ C^S_\ind \ge 1$ or $  \frac{ \kappa}{C^{i}_\ind }\ge 1$   then $k(x) \in [0,1]$. 
In general,
$ 0 \le k(x) \le  \min\big\{  \frac{1}{C^S_\ind}, \frac{\kappa }{C^{i}_\ind} \big\} =M_1$
with $M_1$ as in \eqref{M}. Hence
$$ t(x) = {M_1}^{-1} k(x) = \max \left(C^S_\ind, \frac{ C^{i}_\ind}{\kappa} \right) k(x) \in [0,1].$$
Now we use a result from the Stein-Chen method to assess the distance to a Poisson distribution in total variation distance, Eq.~(1.23), p.~8, from \cite{Barbour_etal1992}. This result states that, if $W$  is the sum of $n$ independent Bernoulli random variables with success probabilities $p_i$ and $\E W = \lambda = \sum_{i=1}^n p_i$ and $Z \sim \Pois(\lambda)$, then
\beam \label{poissonkey} \sup_{h: \mathbb{Z}^+ \rightarrow [0,1] }| \E k(W)  - \E k (Z) | \le \min (1, \lambda^{-1}) \sum_{i=1}^n p_i^2.
\eeam
 Applying \eqref{poissonkey} to the function $t(x)$ and keeping \eqref{zwischenschritt} in mind yields \eqref{bound6}. 
The bound \eqref{bound7} follows similarly. 
 Finally,
 \beao
\min\left\{ \kappa  \frac{A_{ij}^\al}{C^i_\ind}, \frac{A_{kj}^{\alpha}}{C^k_\ind} \right\}
&=& \min\left\{ \kappa \, \frac{\frac{\bone (i\sim j)}{\deg(j)^{\alpha }} }{C^i_\ind},   \frac{\frac{\bone (k\sim j)}{\deg(j)^{\alpha }} }{C^k_\ind} \right\} \\
&=& M_3 \bone (i\sim j)\bone (k\sim j) \Big( 2 + \sum_{\ell \ne i,k} \bone (\ell \sim j)\Big)^{-\alpha}.
\eeao
 As the positive function
$  k (x)=(2+x)^{-\alpha} $ is bounded by 1  and
as $\sum_{\ell=1, \ldots, q; \ell \ne i, k} \mathds{1}(\ell \sim j)$ is a sum of independent Bernoulli variables, 
 \eqref{poissonkey} can be applied, and \eqref{bound5} follows.
 
For the Conditional Tail Expectations,  with \eqref{normkey}, 
\beao
A_{ij} \| A e_j\|^{\alpha - 1}
&=& \Big(\frac1{\deg(j)^{r-1}} \Big)^{\frac{\al-1}{r}}\frac{\bone (i\sim j)}{\deg(j)}
\, = \, \Big(\frac{\bone (i\sim j)}{\deg(j)}\Big)^{\frac{\alpha(r-1) + 1}{r}}    = A_{ij}^{\frac{\alpha(r-1) + 1}{r}}.
\eeao
Hence \eqref{bij} applies, and yields \eqref{bound3}. 
 Similarly, with \eqref{normkey},
 \beao
A_{ij}^{\alpha - 1} \|A e_j\| &=&  \frac{\bone (i\sim j)}{\deg(j)^{\alpha - 1}} \Big(\frac1{\deg(j)^{r-1}} \Big)^{\frac{1}{r}}=  A_{ij}^{\alpha - \frac{\ 1}{r}}.
\eeao
Again \eqref{bij} applies, and yields \eqref{bound4}. \\
For the last part we mimick the proof of Proposition 4.1 in \cite{KKR1}.
By the independence of the edges,
$$\E\big[[A_{kj}^{\alpha-1} A_{ij}\big] =\E\Big[ \bone(i \sim j) \bone (k \sim j) \frac{1}{\deg(j)^\alpha}\Big] =  p_{ij} p_{kj}  \E \Big[ \Big(2 + \sum_{\ell=1, \ldots, q; \ell \ne i, k} \mathds{1}(\ell \sim j)\Big)^{-\alpha} \Big] .$$
  Again 
\eqref{poissonkey} can be applied and the bound \eqref{lastbound} follows.
\eproof

\brem
Using \eqref{bound1} and \eqref{bound2} the constants $M_1, M_2, $ and $M_3$, as well as the expressions on the left-hand side of Proposition \ref{ex41prop}, could be bounded further if desired, by a  straightforward but tedious calculation. 
\erem

\brem
Proposition~\ref{ex41prop} gives an exact bound on the distance to Poisson; no asymptotic regime is suggested. Hence it can be interpreted in different asymptotic regimes.  

If the number $d$ of objects increases, while the number  $q$ of agents  is such that $q = o(d^{\frac12})$, and the number of objects which an agent would connect to, stays constant in expectation, in a fashion so that $p_{ij} \sim \frac{c}{d}$ for a fixed $c$, then$ B(j)$ and $B(i,j,k) $ are  of order $q d^{-2}$, $B(i,j)$ is of order $qd^{-3} $; as long as $q = o(d^{\frac12})$ the Poisson approximation will be suitable. 

Similarly if the number $q$ of agents increases and the number $D$ of objects only increases as $o(q^{\frac12})$ and if  $p_{ij} \sim \frac{c}{q}$ for a fixed $c$, the Poisson approximation would be suitable.
\erem

\subsubsection{Homogeneous independent  bipartite graph model: Illustrations}

To depict our results we consider the most basic case of the bipartite graph that the edges are not only independent but also equally likely; denote the edge probability with $p\in[0,1]$. We also call the edge probability the connectivity parameter, as it is directly proportional to the density of the network. In this network model all agents behave exchangeably.
For this model the market ranges from a market with no activity at all ($p=0$) to a complete graph
($p=1$).
Note that non-conditional risk measures on this type of network have already been studied in \cite{KKR1}.
Here, we are interested in the asymptotic expressions  given in Proposition~\ref{VaRcond} as well as in Theorem~\ref{asymprisk}  concerning the degree of tail dependence, conditional tail expectations and conditional Value-at-Risk, seen as a function of the edge probability $p$, and as a function of $\kappa$ where applicable.
 In all cases, for simplicity of exposition, we concentrate on the interaction of an agent given the systemic stress and vice versa. 
We use the following abbreviations for the right-hand asymptotic expressions of the  conditional systemic risk measures in  \eqref{ExpectedShortfallMulti} and  \eqref{ExpectedShortfallFi}, respectively: (\AS\ corresponds to an agent's risk to exceed its threshold given that the system exceeds its  threshold,  and  \SA\ to the system's  risk to exceed its threshold given that a specific agent exceeds its threshold):
\begin{align*}
&C^{\AS}_{\ind}=C^{\AS}_{\ind}(i)=\frac{\alpha}{\alpha-1}(C^S_{\ind})^{1/\al-1}
\sum_{j=1}^d K_j \E[A_{ij}\|Ae_j\|^{\al-1}] \gamma^{-1/\al},\\
& C^{\SA}_{\ind}=C^{\SA}_{\ind}(i)=\frac{\al}{\al-1}(C^i_{\ind})^{1/\al-1}\sum_{j=1}^{d}K_j \E[ A_{ij}^{\alpha-1}\|Ae_j\|]\gamma^{-1/\al}.
\end{align*}
Plots which depend on $p$ start at $p=0.01$. 

\begin{figure}[ht]
\subfigure{\includegraphics[height=4.5cm,width=0.49\textwidth]{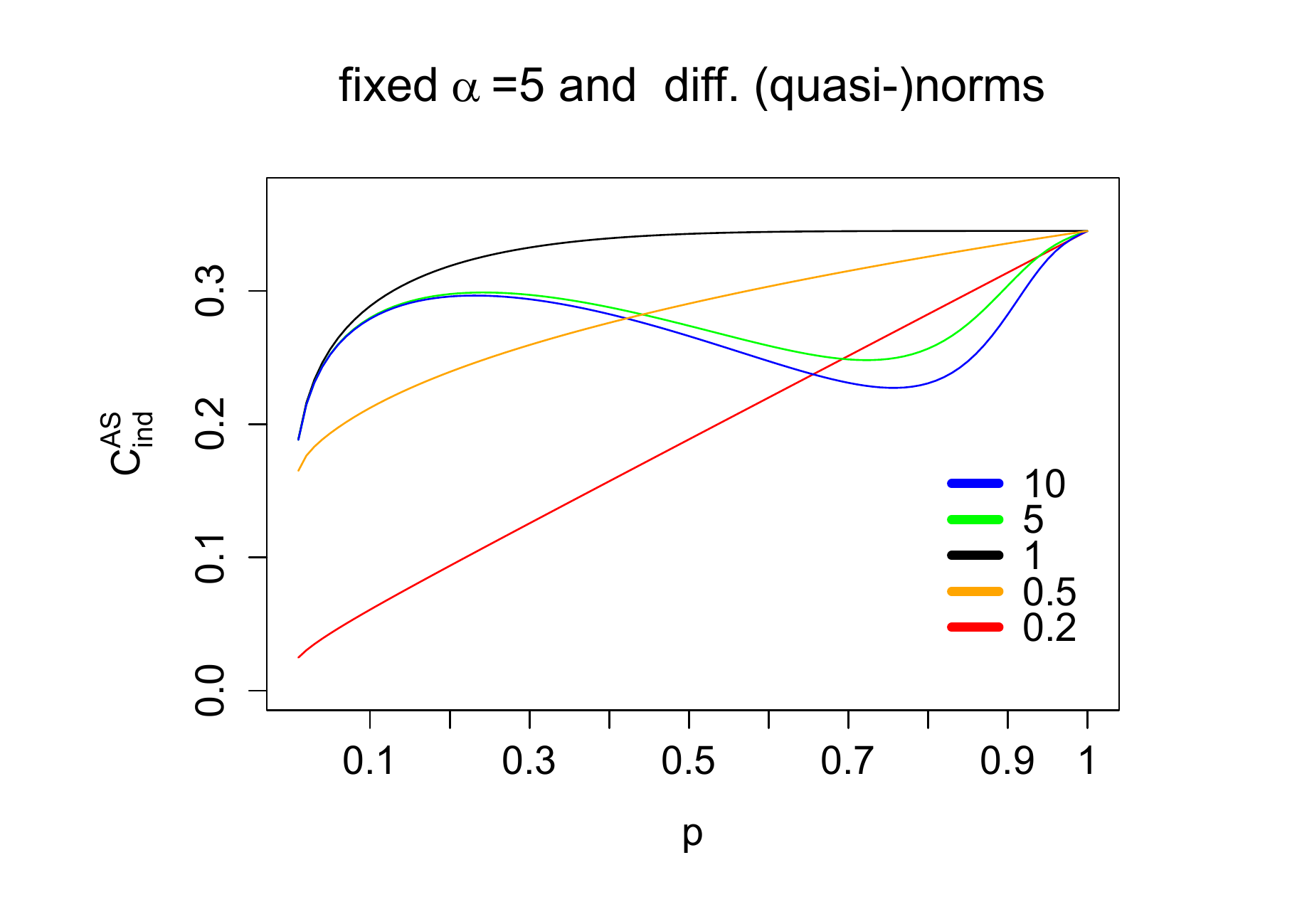}}
\subfigure{\includegraphics[height=4.5cm,width=0.49\textwidth]{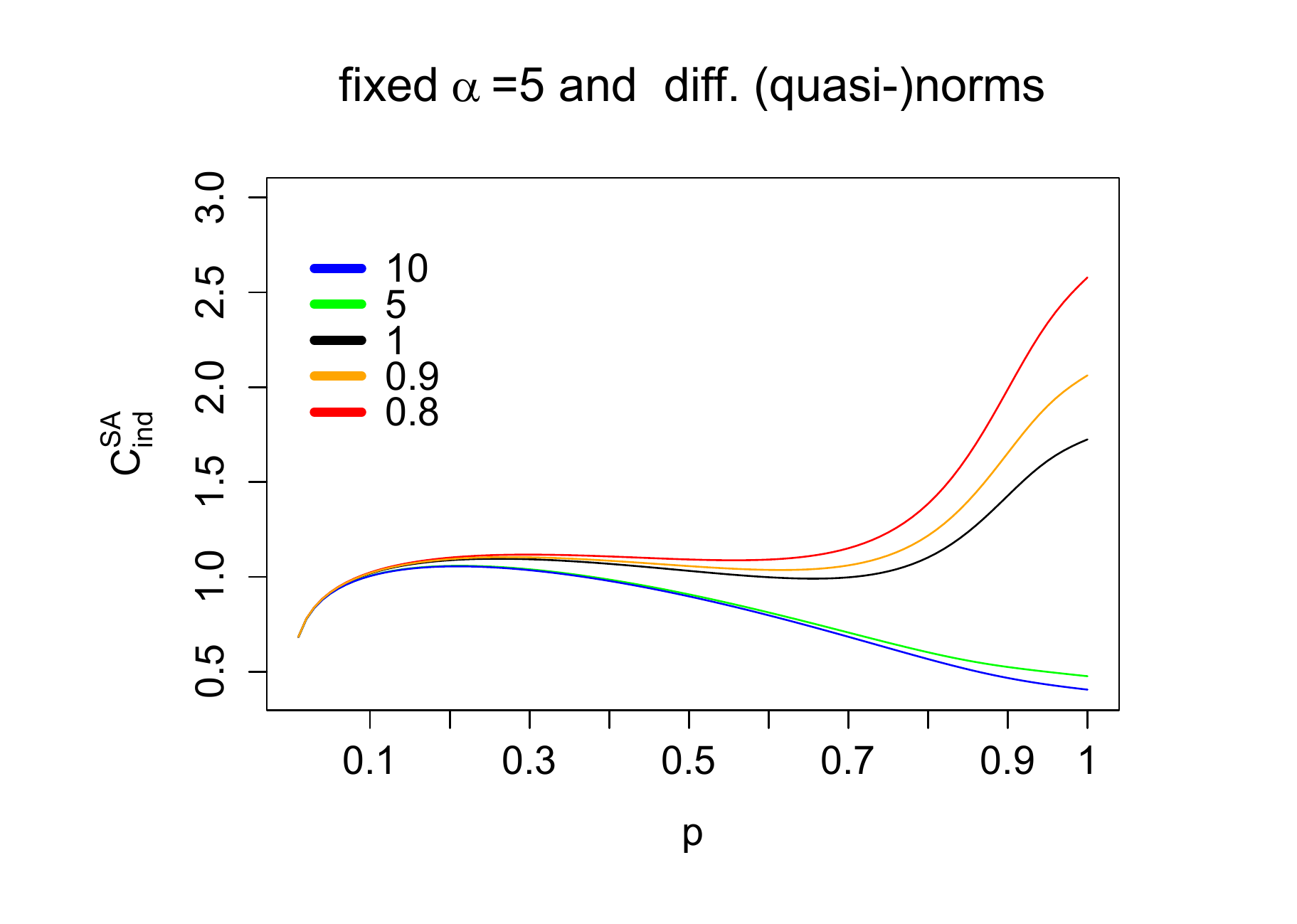}}
\caption{The risk constants $C_{\ind}^{\AS}$ (left)  and $C^{\SA}_{\ind}$ (right) for $\al=5$ (right) as a function of $p$ for different norms and quasi-norms. The plots start with $p=0.01$. Left: for $r > 1$ the curve is non-monotone. while for $r\le 1$ it is monotone increasing. Right: the curves are non-monotone for all values of $r$ considered. 
}
\label{CAS}
\end{figure}

In Figure~\ref{CAS}, both quantities are plotted as functions of the edge probability $p$ exemplarily for different (quasi-)norms while fixing the tail index $\alpha=5$. 
The left-hand plot shows the curves of $C^{\AS}_{\ind}.$ 
As the parameter $p$ increases, the connectivity in the network increases,  having two effects: Firstly,  more object are insured,
hence, the agents take a greater risk  load. Secondly, risk sharing among agents who jointly insure an object increases. We can then clearly recognise that the norms with $r>1$ favour diversification leading to a non-monotone behaviour of the curve as the result of the two compelling characteristics of a greater risk load and positive diversification effects, whereas in the quasi-norm case, diversification is punished and strengthens  the effect of a greater risk load.
The right-hand plot shows the curve of $C_{\ind}^{\SA}. $ Since we consider market losses aggregated by some (quasi-)norm, the losses even for the complete market depend on the norm parameter $r$, which is  one major difference to the left-hand plot. We also recognise the appearance of  non-monotone curves  even for quasi-norms and the sum norm.  In the quasi-norm case there is a (relatively high) value of $p$ at which  the risk constant $C^{\SA}$ has a local minimum.

\begin{figure}[ht]
\subfigure{\includegraphics[height=4.5cm,width=0.49\textwidth]{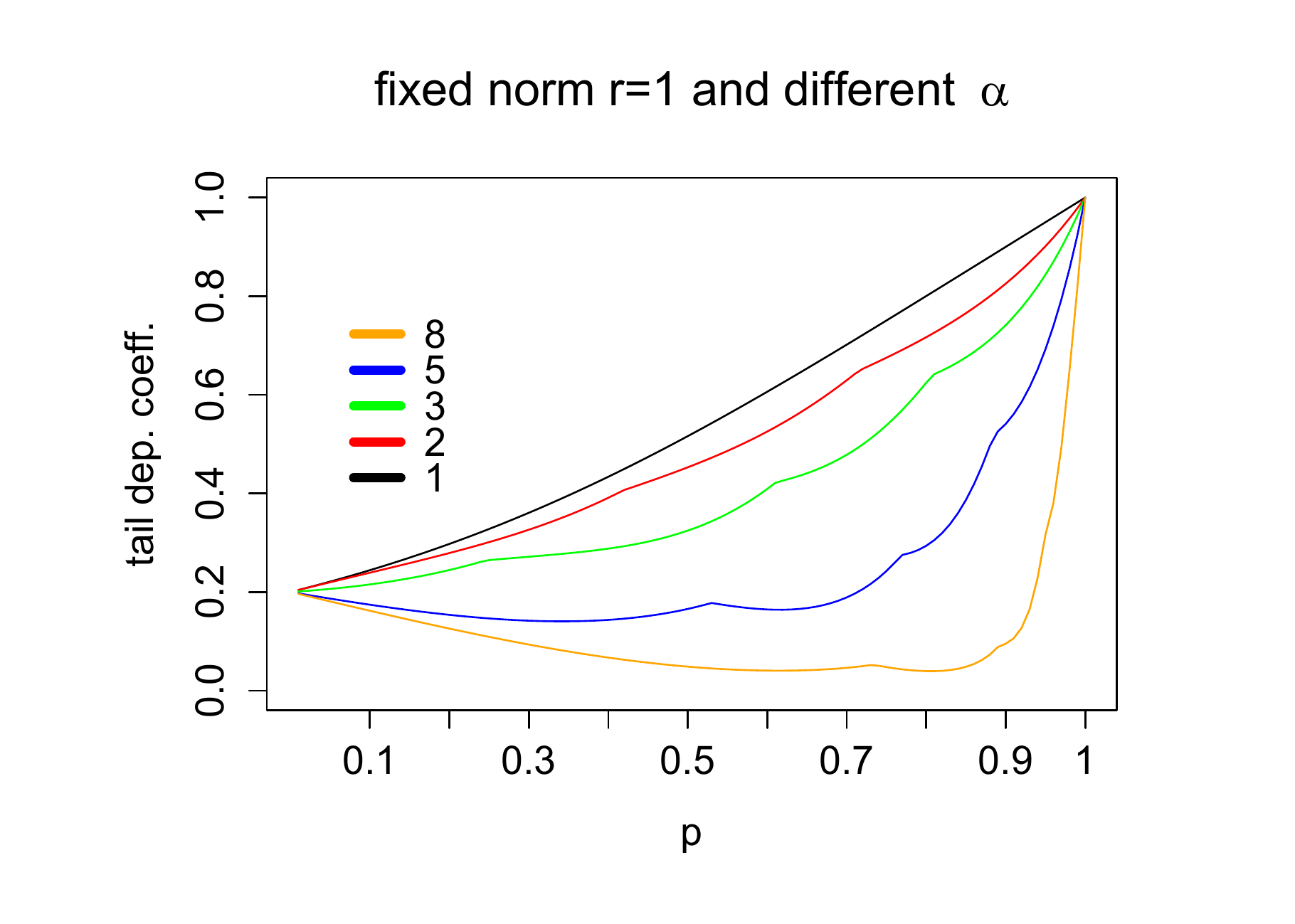}}
\subfigure{\includegraphics[height=4.5cm,width=0.49\textwidth]{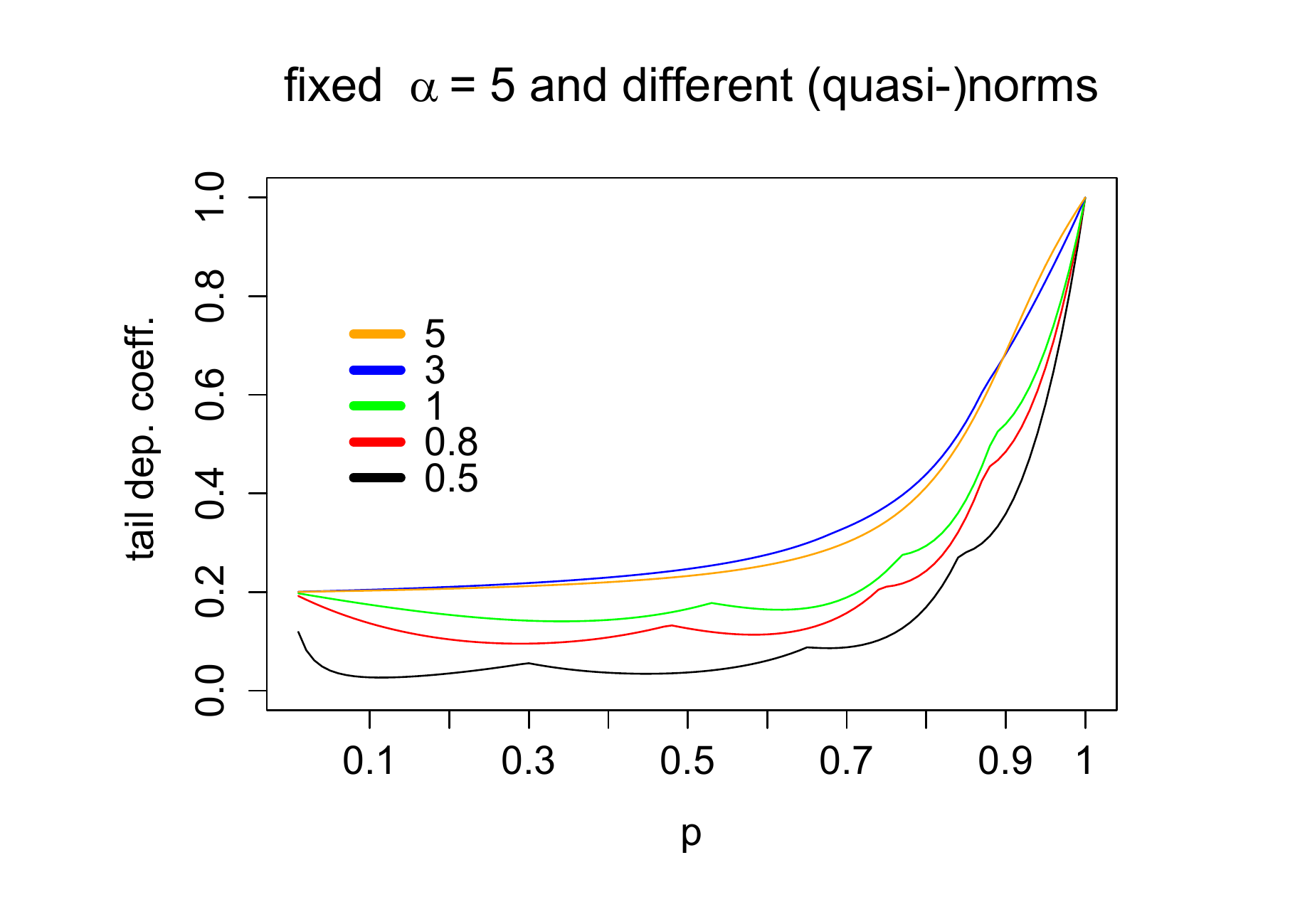}}
\caption{The tail dependence coefficient  from
\eqref{eq:VaR:FigivenFindep} with $\kappa=1$.   Left:  fixed sum-norm ($r=1$) and different tail indices $\alpha$; for small $\alpha$, the tail dependence coefficient  is almost linear, while for larger values of  $\alpha$, the curves are  non-monotone. 
Right: tail index $\alpha=5$ fixed and different (quasi-)norms. 
 The tail dependence coefficient is almost constant before increasing steeply. In both plots
 peaks only appear for the sum-norm and the quasi-norms. }
\label{tail_dep_coeff}
\end{figure}

 Figure~\ref{tail_dep_coeff} illustrates the symmetric tail dependence coefficient, which is \eqref{eq:VaR:FigivenFindep} for $\kappa=1$;i.e.,
$$
\pr{F_i>\VaR_{1-\ga }(F_i)\mid \|F\|> \VaR_{1-\ga}(\|F\|)}
 \to 
\sum_{j=1}^{d} K_j  \E \min\Big\{  \frac{\|Ae_j\|^\al}{C^S_\ind}, \frac{ A_{ij}^{\alpha}}{C^{i}_\ind}\Big\}
$$
as a function of the edge probability $p$. 
In the left-hand plot, which concentrates on the very natural sum-norm ($r=1$) and different tail indices $\alpha$, we observe that for small $\alpha$, the tail dependence coefficient  is almost linear in the network connectivity parameter $p$, but for larger values of  $\alpha$, the behaviour becomes  non-monotone and there exists  a locally optimal connectivity parameter $p$.   As effects of the minimum function, we see small peaks on the curves. In the right-hand plot, fixing $\alpha=5$, the tail dependence coefficient is almost constant before increasing steeply when the network is nearly a complete graph. Furthermore,  peaks, which occur through the minimum function, only appear for the sum-norm and the quasi-norms here.

\begin{figure}[ht]
\subfigure{\includegraphics[height=6cm,width=0.5\textwidth]
{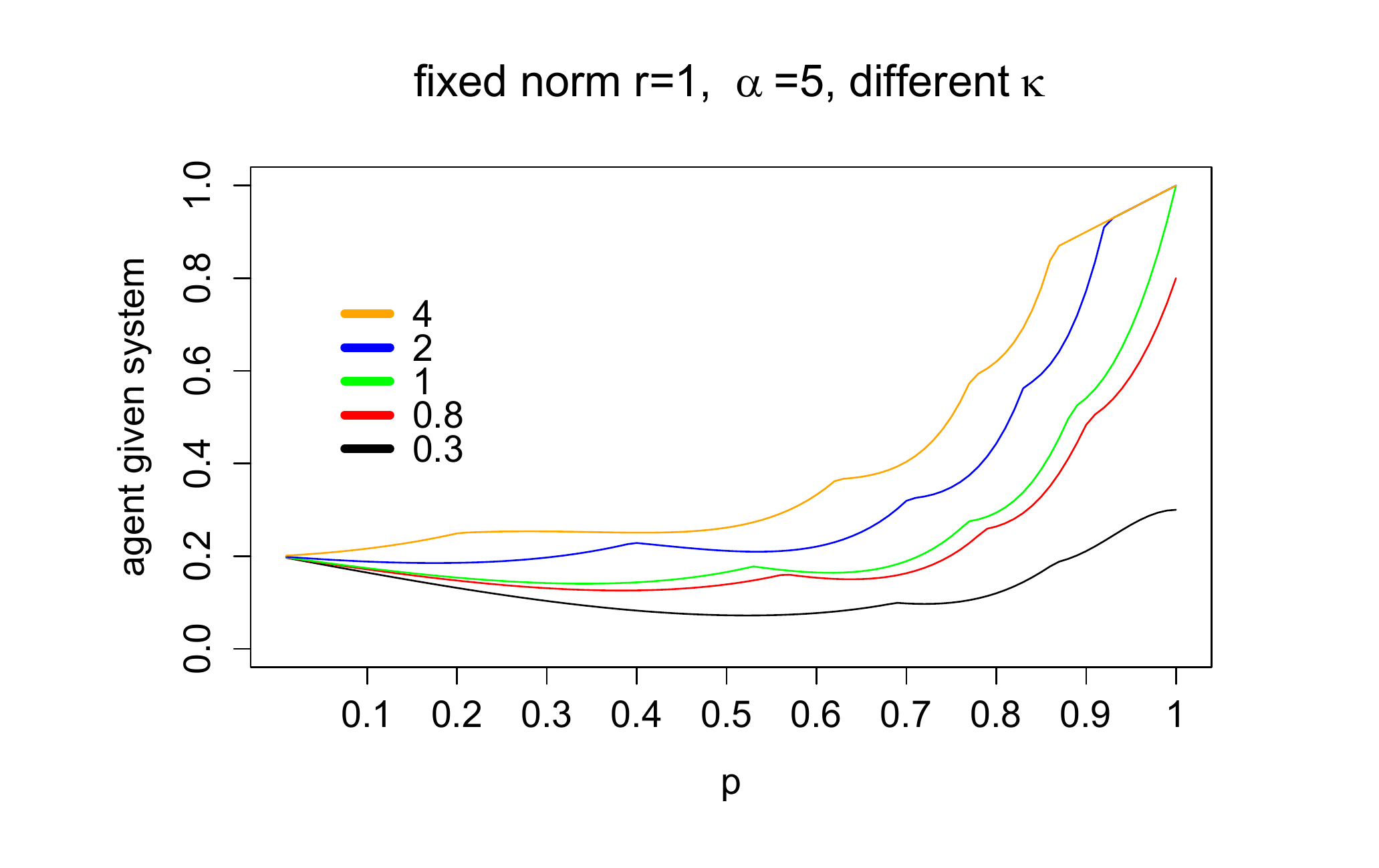}}
\subfigure{\includegraphics[height=6cm,width=0.5\textwidth]
{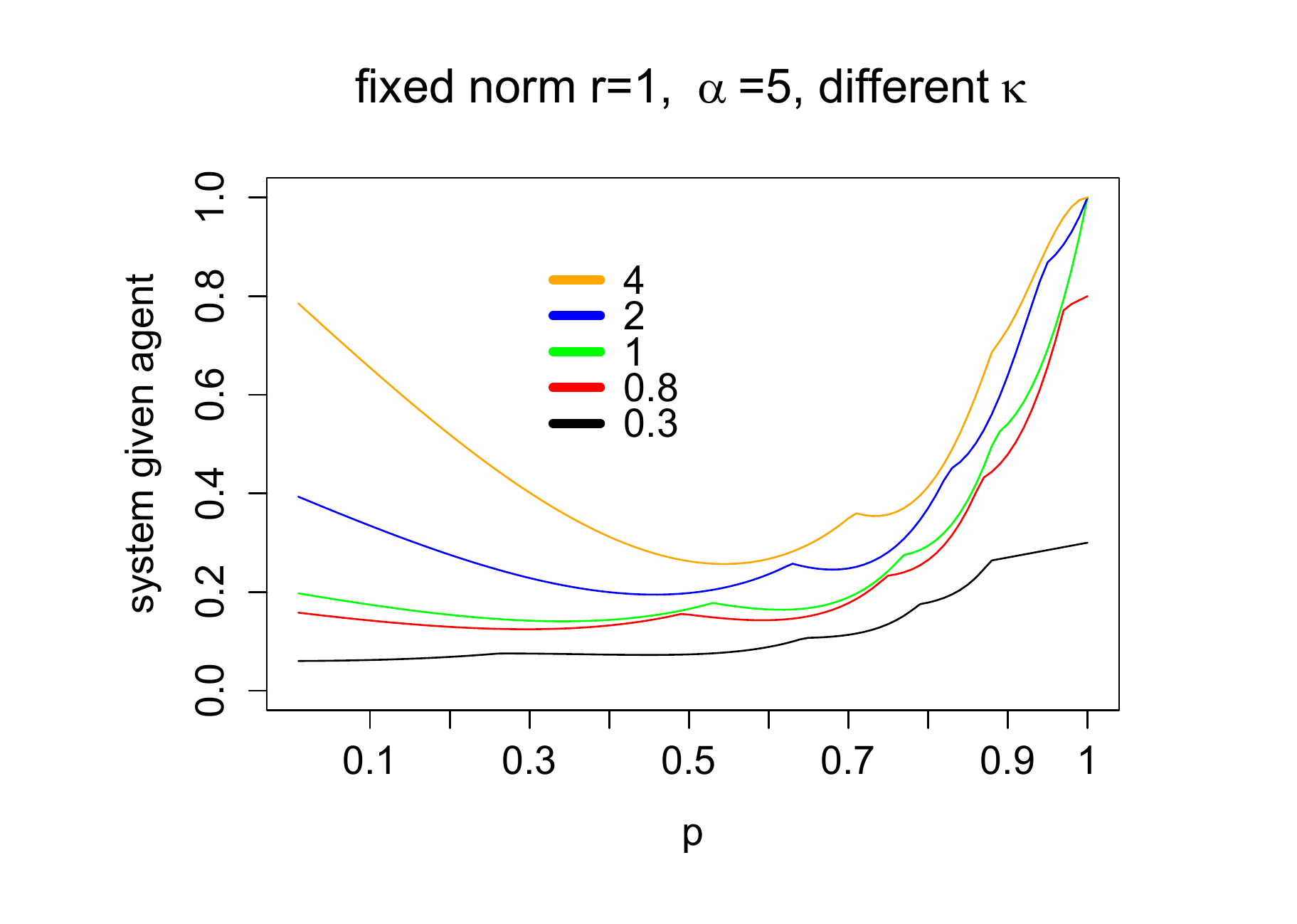}}
\caption{Asymmetric probablities of tail dependence from Proposition~\ref{VaRcond} for different values of $\kappa$, with $\alpha =5$ and taking the sum norm. 
Left: agent given system. As $p \rightarrow 0$, all curves converge to the same value $0.2$. As $p \rightarrow 1$, for $\kappa \ge 1$ the curves merge into one single curve, and for $\kappa < 1$ the curves converge to $\kappa$. 
Right: system given agent. For small $p$ the curves are well separated. For $p \rightarrow 1$ the curves tend to $\min(1, \kappa)$. 
} 
\label{asymtaildep_diffkappa}
\end{figure}

Allowing for asymmetry of the tail dependence coefficient in  Proposition~\ref{VaRcond}  through the value $\kappa$ potentially differing from 0, we illustrate the results versions associated with
agent behaviour conditional on system distress and vice versa; i.e., the right-hand sides  in \eqref{eq:VaR:FigivenFindep} (agent given system) and \eqref{eq:VaR:FgivenFiindep} (system given agent), which  read as
\beao
\pr{F_i>\VaR_{1-\ga \kappa}(F_i)\mid \|F\|> \VaR_{1-\ga}(\|F\|)}
& \to &
\sum_{j=1}^{d} K_j  \E \min\Big\{  \frac{\|Ae_j\|^\al}{C^S_\ind}, \kappa \frac{ A_{ij}^{\alpha}}{C^{i}_\ind}\Big\},\ \gamma \to 0,\\
 \pr{\|F\|> \,\VaR_{1-\kappa \ga}(\|F\|)\mid F_i>\VaR_{1-\ga  }(F_i)}
& \to &
\sum_{j=1}^{d} K_j \E \min\Big\{{\kappa} \frac{\|Ae_j\|^\al}{C^S_\ind},   \frac{A_{ij}^{\alpha}}{C^{i}_\ind} \Big\},\ \gamma \to 0.
\eeao

Figure~\ref{asymtaildep_diffkappa} shows both quantities as functions of the edge probability $p$, exemplarily for the tail index $\alpha= 5$ and the sum norm.
In the left-hand plot---agent given system---all curves apparently converge to the same point as $p\to 0$ which in our case is close to $0.2$. Hence, the influence of $\kappa$ diminishes as the network gets less  connected. As we move to the complete network; i.e. for $p\rightarrow 1$,  for all $\kappa \geq 1,$ the curves converge to one single curve, which is clear from the formulas, and in the case of $\kappa<1,$ the curves converge to $\kappa$, which can  be recognized  as the work of the minimum function. Contrary to the left-hand plot, in the right-hand plot---system given agent---the values for less connected networks; i.e., small values of $p$, lie far apart from each other. This observation can be explained as follows: for small $p$ it would not be unusual to see some $A_{ij}=0$, and it would not be very surprising to see the empty network, with $A=0$. The probability that the norm of $F$ is positive is larger than the corresponding probability of a single explosure $F_i$, hence,  $\kappa$ is multiplied by the factor which has the greater probability to be non-zero. In both pictures,  the curves for $\kappa\neq 1$ are dilated to the left and to the right for $\kappa >1$ and $\kappa<1$, respectively, compared to the symmetric case $\kappa=1$;  the directions of dilation are different in the left and the right plot.

\begin{figure}[ht]
\subfigure{\includegraphics[height=6cm,width=0.5\textwidth]
{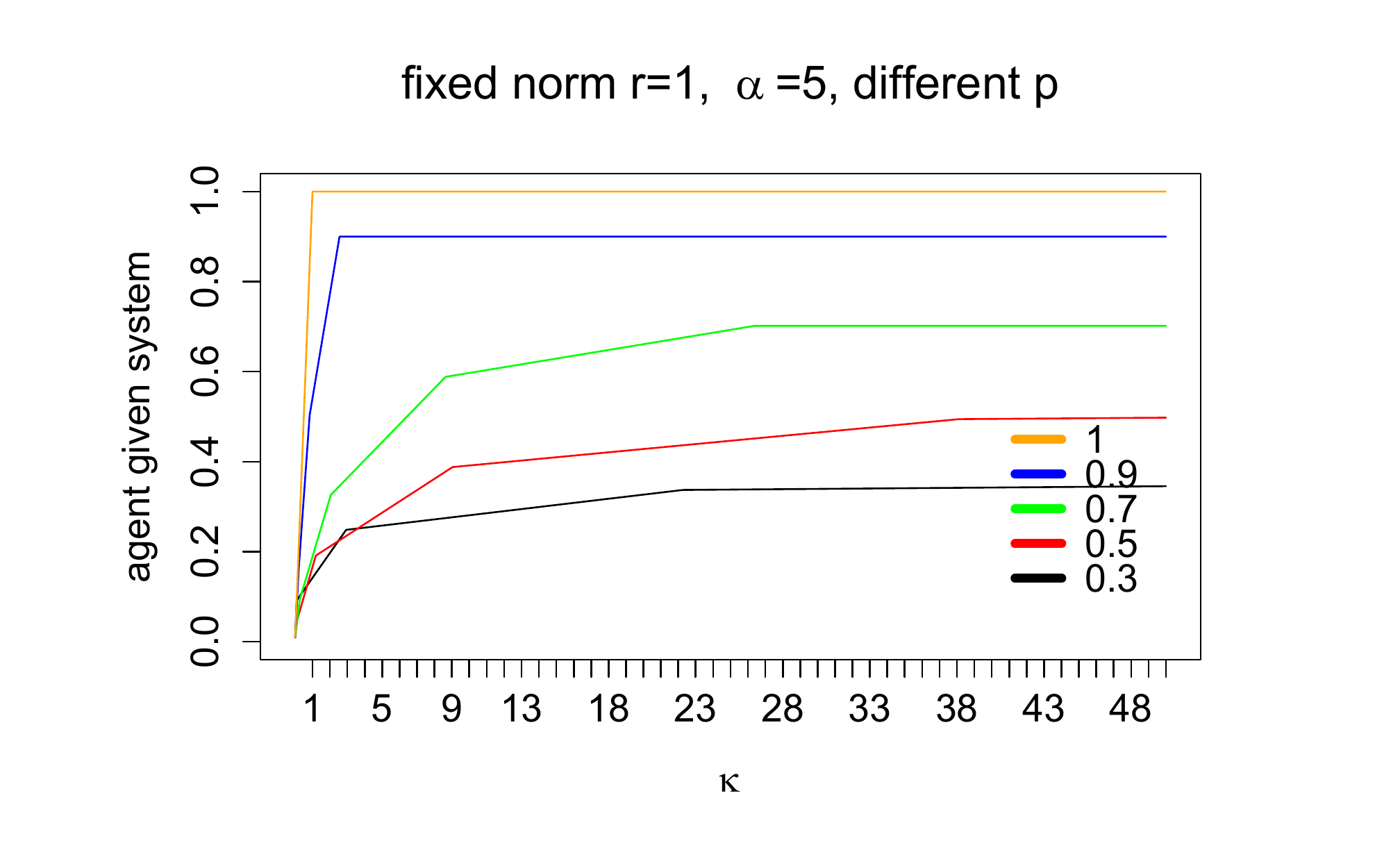}}
\subfigure{\includegraphics[height=6cm,width=0.5\textwidth]
{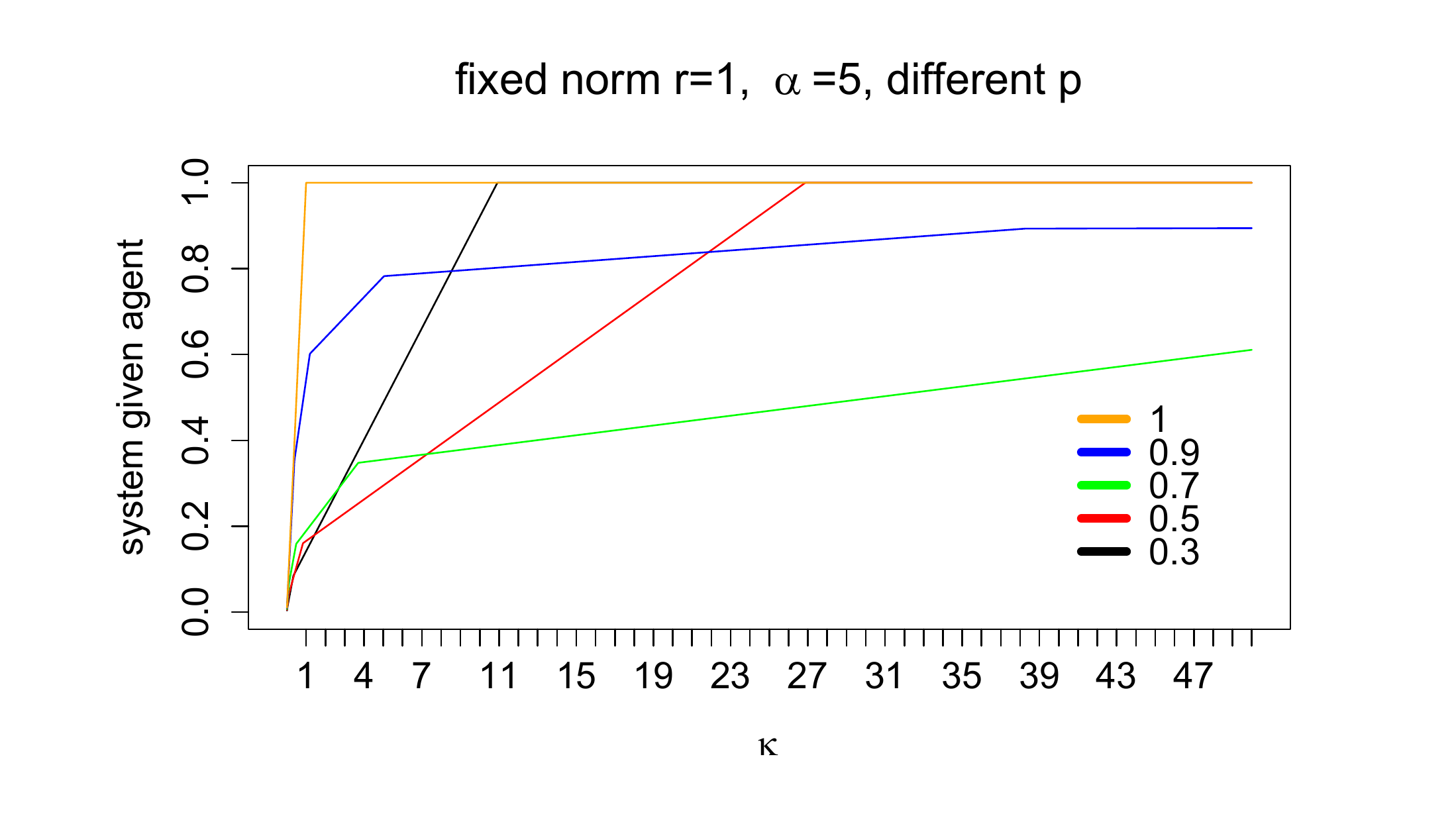}}
\caption{Asymmetric probablities of tail dependence from Proposition~\ref{VaRcond} for different values of  the edge probability  $p$. Left: agent given system. The curves are piecewise linear and converge to a value close to $p$ as $\kappa \rightarrow \infty$. Right: system given agent.  The curves are horizontal at level 1 for $\kappa$ sufficiently large.  }
\label{tail_dep_kappa_bigandrunning}
\end{figure}

\begin{figure}[ht]
\subfigure{\includegraphics[height=6cm,width=0.5\textwidth]
{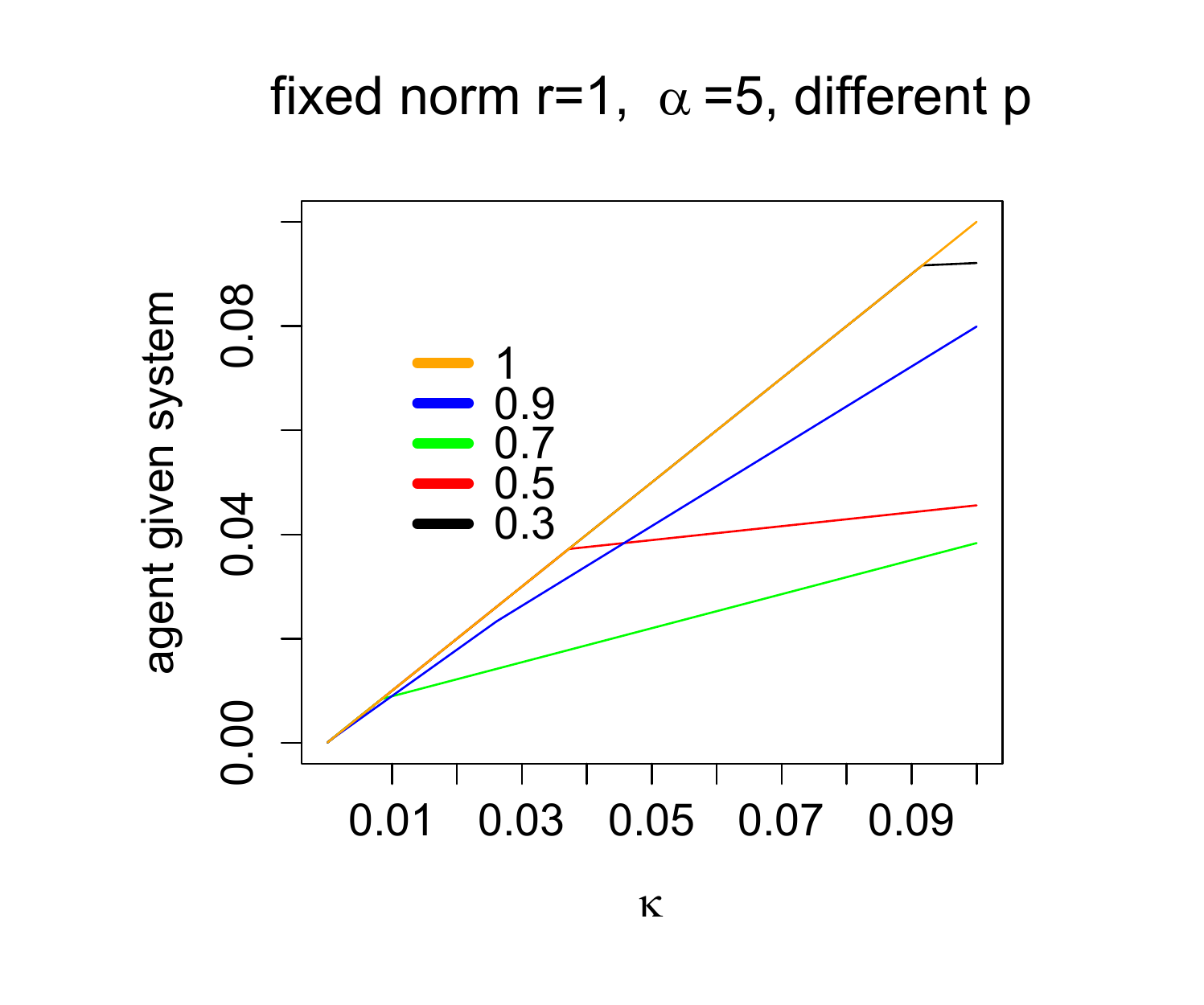}}
\subfigure{\includegraphics[height=6cm,width=0.5\textwidth]
{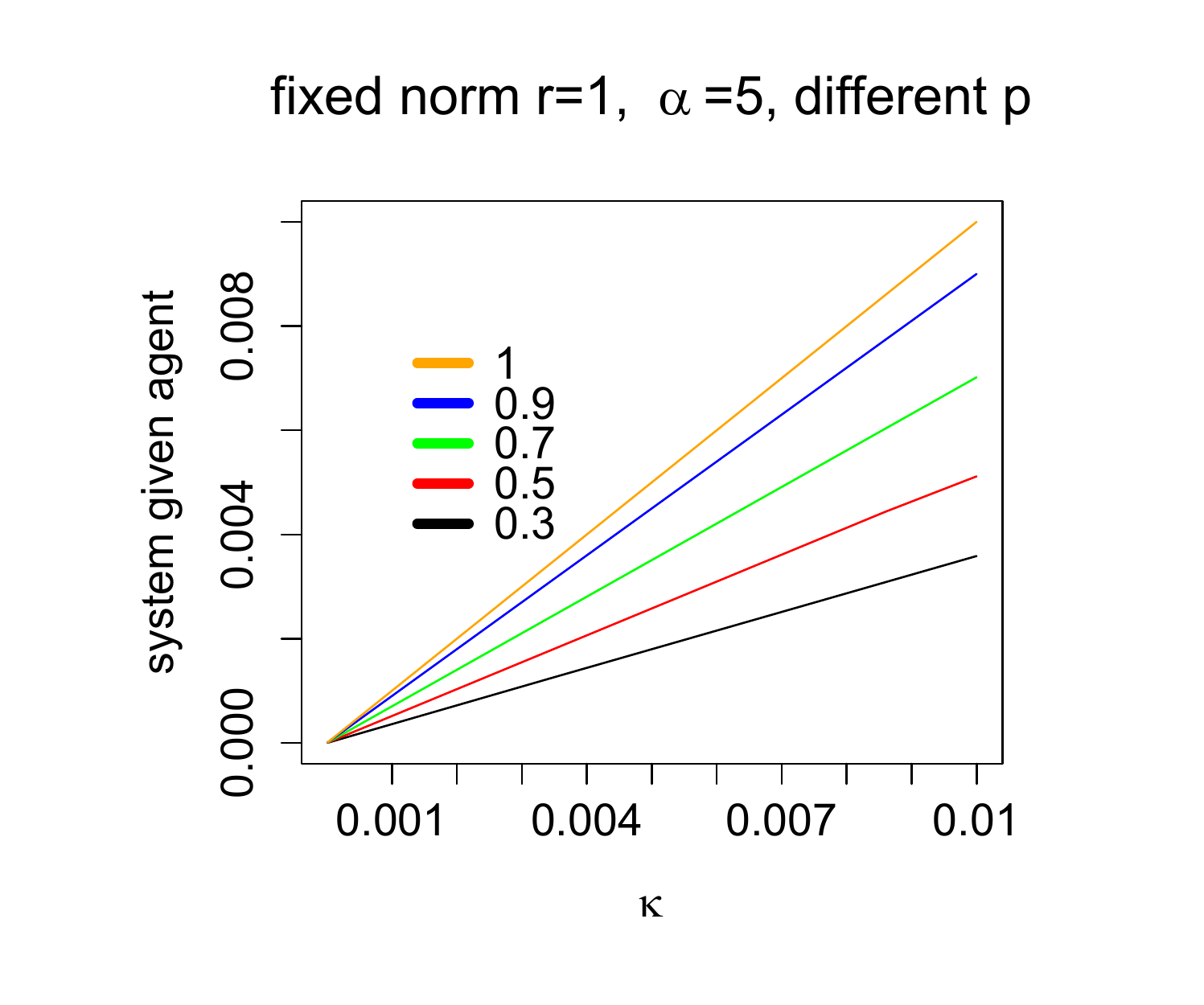}}
\caption{Asymmetric probablities of tail dependence from Proposition~\ref{VaRcond} for different values of  the edge probability $p$. The plot  starts with $\kappa=0.001$ and $\kappa=0.0001$, respectively.  Left: agent given system. For $\kappa$ sufficiently small all curves turn to a slope of 1. Right: system given agent. All slopes are different and close to the respective $p$. 
 }
\label{tail_dep_kappa_smallandrunning}
\end{figure}

Figure~\ref{tail_dep_kappa_bigandrunning} 
studies  the same quantities but now as a function of $\kappa$, with some exemplary values for $p$. The resulting curves are  piecewise linear as a result from taking the expectation. For deterministic matrices one would see only one line with particular slope before turning horizontal. 
For the risk of agent conditioned on system the curves are finally constant at level $p$, while for risk of system given agent the curves are horizontal at level 1 for $\kappa$ sufficiently large.

Figure~\ref{tail_dep_kappa_smallandrunning} depicts 
 the behaviour of theses curves for $\kappa$ near to zero, which is  connected to the asymptotics of the conditional Value-at-Risk in Proposition~\ref{asymprisk}  through Theorem~\ref{help_covar}. 
 The left-hand plot in Figure~\ref{tail_dep_kappa_smallandrunning} shows that for $\kappa$ sufficiently small all curves turn to a slope of 1. 
 This fact is reflected in the asymptotics of $\ICoVaR_{1-\gamma, \gamma_i}(F_{i}\mid \|F\| )$ by the absence of an additional factor. In contrast in the right-hand plot we observe different slopes, close to $p$ in each case. 
 These different slopes enter the formula for the $ \SCoVaR_{1-\gamma_i, \gamma}$ as the $\tau(i)$ from \eqref{taui}; in the homogeneous model $\tau(i)= p$.

\subsection*{Acknowledgements}

CK would like to thank Keble College, Oxford, for their support  through a Senior Research Visitorship. 
GR acknowledges support from EPSRC grant EP/K032402/1 as well as from the Oxford Martin School programme on Resource Stewardship.

\bibliography{bibgesine}

\begin{thebibliography}{10}

\bibitem{CoVar}
T.~Adrian and M.K. Brunnermeier.
\newblock Covar.
\newblock Working Paper 17454, National Bureau of Economic Research, October
  2011.

\bibitem{coherent}
P.~Artzner, F.~Delbaen, J.-M. Eber, and D.~Heath.
\newblock {Coherent Measures of Risk}.
\newblock {\em Mathematical Finance}, 9(3):203--228, 1999.

\bibitem{Barbour_etal1992}
A.D. Barbour, L.~Holst, and S.~Janson.
\newblock {\em Poisson Approximation}.
\newblock Oxford University Press, Oxford, 1992.

\bibitem{BasrakPhD}
B.~Basrak.
\newblock {\em The Sample Autocorrelation Function of Non-Linear Time Series}.
\newblock PhD thesis, Rijksuniversteit Groningen, NL, 2000.

\bibitem{Basrak200295}
B.~Basrak, R.A. Davis, and T.~Mikosch.
\newblock Regular variation of {GARCH} processes.
\newblock {\em Stochastic Processes and their Applications}, 99(1):95 -- 115,
  2002.

\bibitem{Beirlant}
J.~Beirlant, Y.~Goegebeur, J.~Segers, and J.~Teugels.
\newblock {\em {Statistics of Extremes: Theory and Applications}}.
\newblock Wiley Series in Probability and Statistics. Wiley, Chichester, 2006.

\bibitem{Brownlees}
C.~T. Brownlees and R.~Engle.
\newblock Volatility, correlation and tails for systemic risk measurement.,
  2010.
\newblock Working Paper Series, Department of Finance, NYU.

\bibitem{axiomsystemic}
C.~Chen, G.~Iyengar, and C.C. Moallemi.
\newblock An axiomatic approach to systemic risk.
\newblock {\em Management Science}, 59(6):1373--1388, 2013.

\bibitem{DJVZ}
J.~Danielsson, K.~R. James, M.~Valenzuela, and I.~Zer.
\newblock {Model Risk of Risk Models}.
\newblock {\em Available at SSRN 2425689}, 2014.

\bibitem{hannes}
H.~Hoffmann, T.~Meyer-Brandis, and G.~Svindland.
\newblock Risk-consistent conditional systemic risk measures.
\newblock Preprint, University of Munich, Germany, 2014.

\bibitem{Huang}
X.~Huang, H.~Zhou, and H.~Zhu.
\newblock Systemic risk contributions.
\newblock {\em Available at SSRN 1650436}, 2011.

\bibitem{Ibragimov2005}
R.~Ibragimov.
\newblock Portfolio diversification and value at risk under thick-tailedness.
\newblock {\em Quantitative Finance}, 9(5):565--580, 2009.

\bibitem{jouini}
E.~Jouini, M.~Meddeb, and N.~Touzi.
\newblock Vector-valued coherent risk measures.
\newblock {\em Finance and Stochastics}, 8(4):531--552, 2004.

\bibitem{KK3}
O.~Kley and C.~Kl\"uppelberg.
\newblock Bounds for randomly shared risk of heavy-tailed loss factors, 2015.
\newblock arXiv:1503.03726.

\bibitem{KKR1}
O.~Kley, C.~Kl\"uppelberg, and G.~Reinert.
\newblock Systemic risk in a large claims insurance market with bipartite graph
  structure, 2014.
\newblock arXiv:1410.8671v2.

\bibitem{overbeck}
E.~Kromer, L.~Overbeck, and K.A. Zilch.
\newblock Systemic risk measures on general probability spaces.
\newblock {\em Available at SSRN 2268105}, 2013.

\bibitem{sigma}
Swiss Re.
\newblock {\em sigma} study.
\newblock
  \url{http://www.swissre.com/media/news_releases/Insured_losses_from_disasters_below_average_in_2014.html#inline},
  2015.

\bibitem{Resnick2007}
S.I. Resnick.
\newblock {\em Heavy-Tail Phenomena}.
\newblock Springer, New York, 2007.

\bibitem{uninsured}
S.~von Dahlen and G.~von Peter.
\newblock Natural catastrohpes and global reinsurance - exploring the linkages.
\newblock {\em BIS Quarterly Review}, 2012.

\bibitem{ZhuLi}
L.~Zhu and H.~Li.
\newblock Asymptotic analysis of multivariate tail conditional expectations.
\newblock {\em North American Actuarial Journal}, 16(3):350--363, 2012.

\end{thebibliography}
\bibliographystyle{plain}

\end{document}